\documentclass[11pt]{article} 
	\usepackage[margin=1.05 in, bottom=0.8 in]{geometry} 
	\usepackage[english]{babel} 
	\usepackage{authblk}
	\usepackage{amsmath, amssymb, textcomp, amsthm, mathtools,mathrsfs}  
	\usepackage{wasysym} 
	\usepackage{stmaryrd}
	\usepackage{dsfont}
	\usepackage{hyperref}
   	\usepackage{natbib}
	\usepackage[utf8]{inputenc} 
	\usepackage{csquotes}
	\usepackage{footnote,footmisc} 
	\usepackage{setspace} 
	\usepackage{threeparttable}
	\usepackage{tikz}
	\usepackage{pgfplots}
	\usetikzlibrary{shapes.misc}
	\usetikzlibrary{patterns}
	\usepackage{xcolor}
	\usetikzlibrary{decorations.pathreplacing}
	\usetikzlibrary{arrows}
	\usepackage{float} 
	\newcommand{\GG}[1]{}
    
	\usepackage[toc,page]{appendix}
	
	\newtheorem{theorem}{Theorem}
	\newtheorem{lemma}{Lemma} 
	\newtheorem{proposition}{Proposition} 
	\newtheorem{definition}{Definition}

	\newtheorem{assumption}{Assumption}
	\newtheorem{example}{Example}
	\newtheorem*{example*}{Example}
	\usepackage{graphicx}  
	\usepackage{marvosym} 
	\newcommand\independent{\protect\mathpalette{\protect\independenT}{\perp}}
    \def\independenT#1#2{\mathrel{\rlap{$#1#2$}\mkern2mu{#1#2}}}
    \DeclareMathOperator*{\argmin}{\arg\!\min}
    \DeclareMathOperator*{\argmax}{\arg\!\max}
    
    \DeclareMathOperator{\Wtilde}{\widetilde{W}}
    \DeclareMathOperator{\Xtilde}{\widetilde{X}}
    \DeclareMathOperator{\wtilde}{\widetilde{w}}
    \DeclareMathOperator{\xtilde}{\widetilde{x}}
    \newtheorem*{assumptions*}{\assumptionnumber}
\providecommand{\assumptionnumber}{}
\makeatletter
\newenvironment{assumptions}[2]
 {%
  \renewcommand{\assumptionnumber}{Assumption #1#2}%
  \begin{assumptions*}%
  \protected@edef\@currentlabel{#1#2}%
 }
 {%
  \end{assumptions*}
 }
\makeatother
 
\title{Point-identification in multivariate nonseparable triangular models}
\author{Florian F Gunsilius\thanks{This is part of a revised and extended version of my third year paper at Brown University. I thank Susanne Schennach, Ken Chay, Toru Kitagawa, Arthur Lewbel, Adam McCloskey, Andriy Norets, Jesse Shapiro, Simon Freyaldenhoven, and Kevin Proulx as well as the audience at the 2018 Annual Meeting of the Royal Economic Society for very helpful comments which improved the structure of this article. All errors are mine.}}
\affil{Brown University}
\date{\today}
\begin{document}
\maketitle
\begin{abstract}
In this article we introduce a general nonparametric point-identification result for nonseparable triangular models with a multivariate first- and second stage. Based on this we prove point-identification of Hedonic models with multivariate heterogeneity and endogenous observable characteristics, extending and complementing identification results from the literature which all require exogeneity. As an additional application of our theoretical result, we show that the BLP model \citep{berry1995automobile} can also be identified without index restrictions.
\end{abstract}
\section{Introduction}
Over the last two decades several approaches towards identification of nonseparable triangular models of the form
\begin{align*}
Y&=m(X,\varepsilon)\\
X&=h(Z,U)
\end{align*}
have been developed, where $Y$ is the outcome, $X$ is an endogenous regressor, $U$ and $\varepsilon$ are latent error terms, and $m$ and $h$ are unknown production functions. Some results focus on point-identification of the second stage (\citeauthor{d2015identification} \citeyear{d2015identification}, \citeauthor{torgovitsky2015identification} \citeyear{torgovitsky2015identification}) while others identify average or marginal effects (\citeauthor{blundell2003endogeneity} \citeyear{blundell2003endogeneity}, \citeauthor{chesher2003identification} \citeyear{chesher2003identification}, \citeauthor{imbens2009identification} \citeyear{imbens2009identification}, \citeauthor{schennach2012local} \citeyear{schennach2012local}, \citeauthor{matzkin2016independence} \citeyear{matzkin2016independence}). 

All of the results aiming for point-identification of the second stage require a univariate and strictly increasing first- or second stage (in particular \citeauthor{imbens2009identification} \citeyear{imbens2009identification}, \citeauthor{d2015identification} \citeyear{d2015identification}, \citeauthor{torgovitsky2015identification} \citeyear{torgovitsky2015identification}), which limits their practical applicability in settings with general multivariate heterogeneity like the Hedonic model or the BLP model (\citeauthor{matzkin2007heterogeneous} \citeyear{matzkin2007heterogeneous} and \citeauthor{berry2014identification} \citeyear{berry2014identification}, p.~1754). A generalization of identification results to a multivariate setting without strong artificial functional form assumptions is hence important, in particular for bridging the gap between economic and econometric theory. 

We therefore provide a new framework for identification in nonseparable (triangular) models in this article, generalizing the seminal result in \citet{matzkin2003nonparametric}. This allows us to prove point-identification in nonseparable triangular models while allowing for both first- and second stage to be multivariate, generalizing the point-identification results in \citet{torgovitsky2015identification} and \citet{d2015identification}.
As the main application of our theoretical result we provide assumptions for the identification of multi-market Hedonic models with endogenous characteristics, complementing and building on the existing results in \citet{ekeland2004identification}, \citet{heckman2010nonparametric}, and \citet{chernozhukov2014single} who all consider Hedonic models with \emph{exogenous} characteristics. In particular, we provide an answer to the open question in the latter article, asking under which conditions one can nonparametrically identify Hedonic models when observable characteristics are endogenous. In a second application, we also show that the BLP model \citep{berry1995automobile} can be nonparametrically identified without the need to assume that individual heterogeneity can be captured by an index, complementing the seminal result from \citet{berry2014identification}.

The article is structured as follows. In section \ref{litreview} we give a brief overview of the current state of the literature and contrast our approach to other existing approaches. Section \ref{mainidentsection} introduces the main theoretical framework and the general identification result for nonseparable triangular models: we introduce the theoretical framework in section \ref{frameworksubsection}, the main assumptions in section \ref{assumptionssubsection}, and the main result in section \ref{mainresultsubsection}. Section \ref{applications} contains the two applications of the main result: point-identification of the BLP model without index restrictions as well as the main application concerning the point-identification of Hedonic models with endogenous characteristics and multivariate heterogeneity. Section \ref{conclusion} concludes. The appendix contains all proofs for the results in the main paper.

\section{The literature on (point-) identification in nonseparable models}\label{litreview}
In this section we link our result to the literature on nonseparable models in general and nonseparable triangular models in particular while giving an intuitive overview of the standard assumptions in the literature. 

Nonseparable triangular models are an extension of nonseparable models of the form $Y=m(X,\varepsilon)$ with exogenous $X$. The literature on identification in these models is large, with several authors seeking to point-identify the production function $m$  (\citeauthor{matzkin2003nonparametric} \citeyear{matzkin2003nonparametric}, \citeauthor{matzkin2007nonparametric} \citeyear{matzkin2007nonparametric}, \citeauthor{imbens2007nonadditive} \citeyear{imbens2007nonadditive} and references therein) while others predominantly aim for identification of average or marginal effects (\citeauthor{heckman2005structural} \citeyear{heckman2005structural}, \citeauthor{altonji2005cross} \citeyear{altonji2005cross}, \citeauthor{hoderlein2007identification} \citeyear{hoderlein2007identification}, \citeauthor{chernozhukov2007instrumental} \citeyear{chernozhukov2007instrumental}, \citeauthor{hoderlein2009identification} \citeyear{hoderlein2009identification}).
The literature on identification of such models has been growing ever since and has been fruitfully applied to and extended in different scenarios like single-market Hedonic models (\citeauthor{ekeland2004identification} \citeyear{ekeland2004identification}, \citeauthor{heckman2010nonparametric} \citeyear{heckman2010nonparametric}, \citeauthor{chernozhukov2014single} \citeyear{chernozhukov2014single}) or nonlinear Difference-in-Difference models (\citeauthor{athey2006identification} \citeyear{athey2006identification}, \citeauthor{d2013nonlinear} \citeyear{d2013nonlinear}). These results in turn also give rise to many identification results for simultaneous equation models (\citeauthor{matzkin2008identification} \citeyear{matzkin2008identification}, \citeauthor{berry2013identification} \citeyear{berry2013identification}, \citeauthor{blundell2014control} \citeyear{blundell2014control}, \citeauthor{matzkin2015estimation} \citeyear{matzkin2015estimation} and references therein).  

All of these results in one way or another require injectivity of the production function $m$; the standard approach in the majority of approaches is the assumption of strictly increasing and continuous $m$ which in turn also requires that $Y$ and $\varepsilon$ are univariate. 
In this article we argue that while monotonicity is often a reasonable assumption to make, there is a more general assumption which allows for point-identification of the production function in more general settings: measure-preservation of $m$. In fact, we show that the combination of a unique nonparametric structure in combination with measure preservation of $m$ leads to identification results which are the natural extensions of the of the identification results which rely on strictly increasing and continuous production functions. 

Allowing for endogenous $X$ is an important extension of the model, especially for practical purposes as many models of interest take this form. For instance, the BLP model \citep{berry1995automobile} and Hedonic models with multivariate heterogeneity and endogenous characteristics, which are the two applications of our main result; a multivariate point-identification result is therefore important as the point-identification results in \citet{torgovitsky2015identification} and \citet{d2015identification} cannot be applied in these settings as they require a univariate first- and second stage. Our identification result has other possible applications, for example screening models with multidimensional consumer heterogeneity as in \citet{aryal2017identifying}; the latter introduces an identification result for screening models based on \citet{d2015identification} and requires strong high level assumptions on the observable distributions as well as strong functional form assumptions on the second stage. In particular, \citet{aryal2017identifying} assumes the existence of a fixed point for identification, whereas we provide low-level and testable assumptions for a \emph{fixed set} to exist. 

Our general identification result covers four important general aspects which have not or only partially been dealt with in the literature. First, it is the only result for multivariate nonseparable triangular models, and can be applied to a wide variety of settings. Second, we can allow for the most general functional forms on the production functions $m$ and $h$, without being forced to require monotonicity as \citet{torgovitsky2015identification}, \citet{d2015identification}, or \citet{aryal2017identifying}. Third, we prove a new mathematical result which leads to easy-to-check sufficient conditions for identification, generalizing the existing sequencing arguments in \citet{torgovitsky2015identification} and \citet{d2015identification}. Fourth, in our result we can allow for a continuous or discrete (or even binary) instrument $Z$, which can be of lower dimension than $X$, analogous to existing approaches; in our setting, we can allow for lower-dimensional $Z$ even when $m$ and $h$ are truly multivariate functions, and not element-wise monotonic. 

\paragraph{Notation}
The standard measurable space is defined by $(\mathbb{R}^d,\mathscr{B}_{\mathbb{R}^d})$ where $\mathscr{B}_{\mathbb{R}^d}$ is the Borel $\sigma$-algebra. For a random variable $X:\Omega\to \mathbb{R}^d$ in some measure space $(\Omega,\mathscr{A},P)$, the measure $P_X$ is the pushforward measure of $P$ via $X$, i.e.~$P_X(E) = P(X^{-1}(E))$ for every Borel set $E\in\mathscr{B}_{\mathbb{R}^d}$. The corresponding distribution function $F_X:\mathbb{R}^d\to[0,1]$ is defined by $F_X(x) = P_X(X\leq x)$. The support of $X$ is $\mathcal{X}$. Conditional distributions are defined by $F_{X|Z=z,W=w}$ with support $\mathcal{X}_{zw}$.
The standard partial order on $\mathbb{R}^d$ induced by the positive cone for $x,x'\in\mathbb{R}^d$ is defined as $x\leq x'$ if and only if $x_i\leq x_i'$ for $i=1,\ldots,d$. Based on this the distribution function $F_X$ is strictly increasing if $F_{X|Z=z_i}(x)< F_{X|Z=z_i}(x')$ whenever $x<x'$ in the standard partial order. Throughout, whenever we require the functions $X = h(Z,U)$ and $Y=m(X,\varepsilon)$ to be invertible, we always mean invertibility with respect to the \emph{second} argument, i.e.~between $X$ and $U$ as well as $Y$ and $\varepsilon$, never with respect to $Z$.

\section{Theoretical section: point-identification of multivariate nonseparable triangular models}\label{mainidentsection}\label{torgovitskysubsection}
This is the main theoretical section where we introduce the general framework for the identification of multivariate nonseparable models as well as the main theorem, which generalizes the seminal results of \citet{torgovitsky2015identification} and \citet{d2015identification} to multivariate nonseparable triangular models. 

\subsection{The general framework for point-identification}\label{frameworksubsection}
Let us start this section with a general outline of the underlying idea for the identification result.
We work within the following model throughout:
\begin{align}\label{triangular}
\begin{split}
Y &= m(X,\varepsilon) \\
X &= h(Z,U),
\end{split}
\end{align}
where the covariate $X$ is endogenous, i.e.~depends on the unobservable error term $\varepsilon$. $U$ is the unobservable and independent error term in the first stage. $m$ and $h$ are unobservable. The variable $Z$ is an instrument in this model. Throughout, we have to assume that $\varepsilon$ is of the same dimension as $Y$ and $U$ is of the same dimension as $X$. The intuitive reason for this is that we need the production functions $m$ and $h$ to be invertible in $\varepsilon$ and $U$ in order to derive at our point-identification result.\footnote{This is the main difference to other attempts in the literature like \citet{kasy2014instrumental}, who sought to identify the nonseparable triangular model whilst allowing for a possibly infinite-dimensional unobservable error term of the first stage. We do need to make a dimensionality restriction for our result to hold.}
\begin{assumption}[Dimensions]\label{dimensionsass}
The supports of $Y, \varepsilon$, $X$, $U$, and $Z$ satisfy $\text{dim}(\mathcal{Y}) = \text{dim}(\mathcal{E}) = d$, $\text{dim}(\mathcal{X}) = \text{dim}(\mathcal{U}) = k$, and $\mathcal{Z}\subseteq\mathbb{R}^m$ for finite integers $d,k,m$.
\end{assumption}
We allow for the whole or part of the vector $X\in\mathbb{R}^k$ to be endogenous. In the case where only a part of $X$ is endogenous, it is customary to write the second stage of \eqref{triangular} as $Y=m(X,W,\varepsilon)$, where $W$ are the exogenous covariates, and condition all results on $W$. In this case the dimension of $U$ has to be reduced to match the dimension of the endogenous variables. For the sake of conciseness we consider all elements of $X$ endogenous and suppress $W$ throughout. 

The idea to allow for multivariate first- and second stages is to find a sufficiently general functional requirement for $m$ and $h$, which seminal results like \citet{imbens2009identification}, \citet{torgovitsky2015identification}, and \citet{d2015identification} require to be strictly increasing and continuous. The issue is that there is no complete order in higher dimensions, so that those classical identification results based on \citet{matzkin2003nonparametric} are not applicable in this setting. We therefore argue that the appropriate generalization of a strictly increasing and continuous function in these settings is a \emph{measure preserving isomorphism}.
\begin{definition}
A map $T:\mathcal{E}\to\mathcal{Y}$ transporting a probability measure $P_\varepsilon$ onto another probability measure $P_Y$ is \emph{measure preserving} if it is measurable\footnote{Measurability of $T$ means that $\mathscr{B}_{\mathbb{R}^d}=T^{-1}\mathscr{B}_{\mathbb{R}^d}$. $\mathcal{E}$ denotes the support of $\varepsilon$.} and 
\begin{equation}
P_Y(E)=P_\varepsilon(T^{-1}(E))
\end{equation}
for every set $E$ in the Borel $\sigma$-algebra $\mathscr{B}_{\mathbb{R}^d}$ corresponding to $Y$.\footnote{$T^{-1}(E)$ denotes the set of points $e\in \mathcal{E}$ such that $T(e)\in E$.} 
If $T$ is invertible and its inverse is also measure preserving, it is called a measure-preserving isomorphism.
\end{definition}

A way to check whether a transformation is measure preserving is by checking this property for all half-open rectangles of the form $(a,b]$, $a,b\in\mathbb{R}^d$.\footnote{Half-open rectangles in $\mathbb{R}^d$ are the $d$-fold cartesian product of half-open intervals, i.e.~$(a,b]\coloneqq \bigtimes_{i=1}^d (a_i,b_i]$, $a_i,b_i\in\mathbb{R}$ and $a = (a_1,\ldots,a_d)'$, $b=(b_1,\ldots,b_d)'$.}
\begin{proposition}\label{cdfmeasurelemma}
$T:\mathcal{E}\to\mathcal{Y}$ with $y=T(\varepsilon)$ transporting $P_\varepsilon$ onto $P_Y$ is \emph{measure preserving} if and only if 
\begin{equation}\label{cdfmeasureeq}
P_Y((a,b]) = P_\varepsilon(T^{-1}((a,b])).
\end{equation}
\end{proposition}
\begin{proof}
This immediately follows from Theorem A.8 in \citet{einsiedler2013ergodic} and the fact that all rectangles of the form $(a,b]$ form a semi-ring in $\mathscr{B}_Y$ and $\mathscr{B}_X$.
\end{proof}

The set of all measure preserving isomorphisms between two distribution functions is large. In particular, notice that a strictly increasing and continuous $T$ in the univariate case must be measure preserving since 
\begin{equation}\label{monotonemeasurepreserving}
P_Y((-\infty,y])\equiv P(Y\leq y) = P(T(\varepsilon)\leq y) = P(\varepsilon\leq T^{-1}(y)) \equiv P_\varepsilon(T^{-1}((-\infty, y]),
\end{equation}
where the third equality follows from the fact that $T$ is continuous and strictly increasing. Strictly increasing and continuous functions are special measure preserving maps because they map every interval of the form $(-\infty,y_1]$ to an interval \emph{of the same form} $T^{-1}((-\infty,y_1])\equiv(-\infty,\varepsilon_1]$ such that both intervals have the same probability (Figure \ref{currentusedplot}). This requirement makes strictly increasing and continuous functions \emph{unique} in the class of measure preserving isomorphisms\footnote{A strictly increasing and continuous function is invertible.} between two \emph{fixed} distributions $F_X$ and $F_Y$, a property which has first been exploited in \citet{matzkin2003nonparametric}. On the other hand, measure preservation only requires that an interval $(-\infty,y_1]$ gets mapped to \emph{some} combination of intervals, a much weaker restriction (Figure \ref{genmeasurepreservplot}).\footnote{Completely formally, the image of a measure preserving set need not even be an interval, but could be more general, like a Cantor set for example.} 

\begin{figure}[h]
\centering
\begin{tikzpicture}
\draw[->, thick] (0,0) to (0,3);
\draw[->,thick] (0,0) to (3,0);
\draw[-,dashed] (0,2.7) to (3,2.7);
\draw[-,dashed] (6,2.7) to (9,2.7);
\draw[->, thick] (6,0) to (9,0);
\draw[->,thick] (6,0) to (6,3);
\draw[-,thick] (0,0) to [out=10,in=200] (2.7,2.7);
\draw[thick] (6,0) to [out=80,in=180] (8.7,2.7);
\draw[-,thick] (0,2.7) to (-0.1,2.7);
\draw[-,thick] (6,2.7) to (5.9,2.7);
\draw node[left] at (0,2.7) {$1$};
\draw node[left] at (6,2.7) {$1$};
\draw node[right] at (3,0) {$\varepsilon$};
\draw node[right] at (9,0) {$Y$};

\begin{scope}
     \clip (0,0) rectangle (1.96,2.5);
     \fill[gray!50] (0,0) to [out=10,in=200] (2.7,2.7) -- (2.7,2.7) -- (2.7,0) -- (0,0);
\end{scope}
\begin{scope}
\clip (0,0) rectangle (1,2.5);
     \fill[pattern=north east lines] (0,0) to [out=10,in=200] (2.7,2.7) -- (2.7,2.7) -- (2.7,0) -- (0,0);
\end{scope}
\draw[-,thick] (1,0) to (1,0.66);
\draw[-,thick] (1.96,0) to (1.96,2.2);
\begin{scope}
\clip (0,0) rectangle (7.2,2.7);
\fill[gray!50] (6,0) to [out=80,in=180] (8.7,2.7) -- (8.7,2.7) -- (8.7,0) -- (6,0);
\end{scope}
\begin{scope}
\clip (0,0) rectangle (6.16,0.66);
\fill[pattern=north east lines] (6,0) to [out=80,in=180] (8.7,2.7) -- (8.7,2.7) -- (8.7,0) -- (6,0);
\end{scope}
\draw[-,thick] (6.16,0) to (6.16,0.66);
\draw[-,thick] (7.2,0) to (7.2,2.2);
\draw[-,thick] (1.96,0) to (1.96,-0.1);
\draw[-,thick] (0,2.2) to (-0.1,2.2);
\draw[-,thick] (6.16,0) to (6.16,-0.1);
\draw[-,thick] (7.2,0) to (7.2,-0.1);
\draw[-,thick] (6,2.2) to (5.9,2.2);
\draw node[below] at (6.175,-0.2) {$a$};
\draw node[below] at (7.2,-0.1) {$b$};
\draw node[below] at (2.1,-0.1) {\footnotesize{$T^{-1}(b)$}};
\draw node[below] at (0.8,-0.1) {\footnotesize{$T^{-1}(a)$}};
\draw node[left] at (6,2.2) {\footnotesize{$F_Y(b)$}};
\draw node[left] at (0,2.2) {\footnotesize{$F_\varepsilon(T^{-1}(b))$}};
\draw node[left] at (6,0.66) {\footnotesize{$F_Y(a)$}};
\draw node[left] at (0,0.66) {\footnotesize{$F_\varepsilon(T^{-1}(a))$}};
\draw[-,thick] (1,0) to (1,-0.1);
\draw[-,dashed] (0,0.66) to (1,0.66);
\draw[-,dashed] (6,0.66) to (6.16,0.66);
\draw[-,dashed] (0,2.2) to (1.96,2.2);
\draw[-,dashed] (6,2.2) to (7.2,2.2);
\end{tikzpicture}
\caption{Map currently used for identification}\label{currentusedplot}
\end{figure}
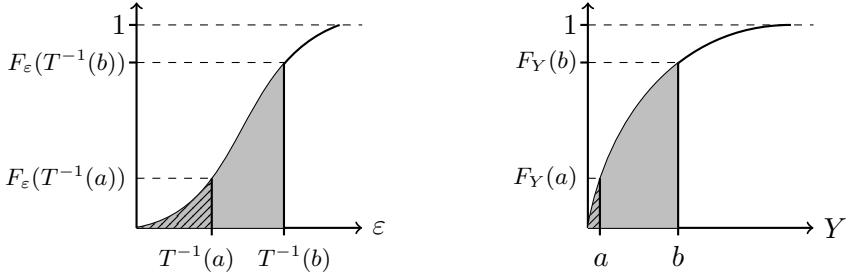

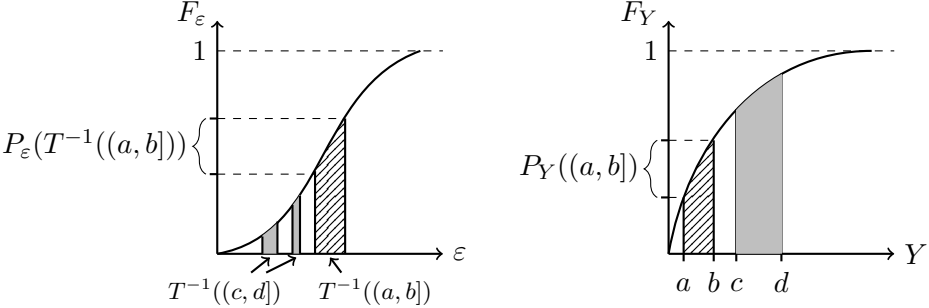
\begin{figure}[h]
\centering
\begin{tikzpicture}
\draw[->, thick] (0,0) to (0,3.1);
\draw[->,thick] (0,0) to (3,0);
\draw[->, thick] (6,0) to (9,0);
\draw[->,thick] (6,0) to (6,3.1);
\draw[-,thick] (0,0) to [out=10,in=200] (2.7,2.7);
\draw[thick] (6,0) to [out=80,in=180] (8.7,2.7);
\draw[-,dashed] (0,2.7) to (3,2.7);
\draw[-,dashed] (6,2.7) to (9,2.7);
\draw[-,dashed] (0,1.06) to (1.3,1.06);
\draw[-,dashed] (0,1.8) to (1.7,1.8);
\draw node[right] at (3,0) {$\varepsilon$};
\draw node[right] at (9,0) {$Y$};
\draw node[left] at (0,3.2) {$F_\varepsilon$};
\draw node[left] at (6,3.2) {$F_Y$};
\draw node[left] at (0,2.7) {$1$};
\draw node[left] at (6,2.7) {$1$};
\draw [decorate,decoration={brace,amplitude=5pt},rotate=0] (-0.15,1.06) -- (-0.15,1.8);
\draw node[left] at (-0.25,1.45) {$P_\varepsilon(T^{-1}((a,b]))$};
\begin{scope}
     \clip (1,0) rectangle (1.1,2) -- (0.6,0) rectangle (0.8, 0.5);
     \fill[gray!50] (0,0) to [out=10,in=200] (2.7,2.7) -- (2.7,2.7) -- (2.7,0) -- (0,0);
\end{scope}
\begin{scope}
    \clip (1.3,0) rectangle (1.7,2.7);
    \fill[pattern=north east lines] (0,0) to [out=10,in=200] (2.7,2.7) -- (2.7,2.7) -- (2.7,0) -- (0,0);
\end{scope}
\draw[-,thick] (-0.1,1.06) to (0,1.06);
\draw[-,thick] (-0.1,1.8) to (0,1.8);
\draw[-,thick] (1.3,0) to (1.3,1.1);
\draw[-,thick] (1.7,0) to (1.7,1.8);
\draw[-,thick] (0.6,0) to (0.6,0.23);
\draw[-,thick] (0.8,0) to (0.8,0.43);
\draw[-,thick] (1,0) to (1,0.65);
\draw[-,thick] (1.1,0) to (1.1,0.77);
\draw[-,thick] (6.9,0) to (6.9,1.925);
\draw[-,thick] (7.5,0) to (7.5,2.4);
\begin{scope}
\clip (6.9,0) rectangle (7.5,2.7);
\fill[gray!50] (6,0) to [out=80,in=180] (8.7,2.7) -- (8.7,2.7) -- (8.7,0) -- (6,0);
\end{scope}
\begin{scope}
\clip (6.2,0) rectangle (6.6,2.8);
\fill[pattern=north east lines] (6,0) to [out=80,in=180] (8.7,2.7) -- (8.7,2.7) -- (8.7,0) -- (6,0);
\end{scope}
\draw[-,thick] (6.2,0) to (6.2,0.75);
\draw[-,thick] (5.9,1.51) to (6,1.51);
\draw[-,dashed] (6,0.75) to (6.2,0.75);
\draw[-,dashed] (6,1.51) to (6.6,1.51);
\draw[-,thick] (5.9,0.75) to (6,0.75);
\draw[-,thick] (6.6,0) to (6.6,1.51);
\draw[-,thick] (6.2,0) to (6.2,-0.1);
\draw[-,thick] (6.6,0) to (6.6,-0.1);
\draw node[below] at (6.2,-0.2) {$a$};
\draw node[below] at (6.6,-0.1) {$b$};
\draw[-,thick] (6.9,0) to (6.9,-0.1);
\draw[-,thick] (7.5,0) to (7.5,-0.1);
\draw node[below] at (6.9,-0.2) {$c$};
\draw node[below] at (7.5,-0.1) {$d$};
\draw node[below] at (0.1,-0.2) {\footnotesize{$T^{-1}((c,d])$}};
\draw[->,thick] (0.45,-0.25) to (0.7,-0.05);
\draw[->,thick] (0.65,-0.25) to (1.05,-0.05);
\draw node[below] at (2.1,-0.2) {\footnotesize{$T^{-1}((a,b])$}};
\draw[->,thick] (1.65,-0.25) to (1.5,-0.05);
\draw [decorate,decoration={brace,amplitude=5pt},rotate=0] (5.85,0.75) -- (5.85,1.51);
\draw node[left] at (5.75,1.15) {$P_Y((a,b])$};
\end{tikzpicture}
\caption{General measure preservation}\label{genmeasurepreservplot}
\end{figure}

To make these concepts more intuitive consider $\mathcal{E}$ as well as $\mathcal{Y}$ to be a continuum of individuals, respectively. The probability measure $P_\varepsilon$ then gives the ``size'' of each (Borel-) subset $E\subset\mathcal{E}$ of individuals, and analogous for $P_Y$. The classical idea for identification using monotonicity is that---in one dimension---requiring the production function $f:\mathcal{E}\to\mathcal{Y}$ to be strictly increasing and continuous means that $f$ \emph{preserves the ordering} of individuals when mapping from $\varepsilon$ to $Y$. This immediately implies that the map $f$ preserves measure, as a group of people $E\subset\mathcal{E}$ with size $P_\varepsilon(E)$ gets mapped to a group of people $f(E)\subset\mathcal{Y}$ of size $P_Y(f(E))=P_\varepsilon(E)$, simply by the fact that the order of individuals needs to be preserved, which follows from \eqref{monotonemeasurepreserving}. Since a strictly increasing and continuous $f$ is invertible, this is our characterization of measure preservation. 
As will become clear, the main requirement for general identification is measure preservation; if $m$ is to be point-identified, we in addition need to require $m$ to be \emph{unique}. Uniqueness comes from functional form restrictions like (multivariate generalizations of) monotonicity, but monotonicity is simply a sufficient condition for measure preservation and uniqueness, which in turn are a sufficient condition for point-identification of $m$. 

On the outset it may seem like a tautology that we require uniqueness of the production function in order to obtain point-identification. Note, however, that there is a distinction between uniqueness and statistical identification, and that the latter does not follow from the former in general. A production function which is theoretically unique need not be identifiable; in fact our identification result below provides the machinery to go from uniqueness to identification in the setting of nonseparable triangular models. In particular, we briefly show below that a linear first-stage relationship of the form $X = \beta Z+U$ with $\mathcal{U}=\mathbb{R}^k$ cannot be identified by our methods for $\beta\neq 0$, which is perfectly analogous to a comment in \citet{torgovitsky2015identification}.

In this respect note that monotonicity and continuity are also structural assumptions on $m$ which ensure its uniqueness and measure-preservation. It is in this sense that our framework generalizes the framework in \citet{matzkin2003nonparametric} as we allow for more general nonparametric function classes than just (univariate) strict monotonicity and continuity. 
So what we require is simply \emph{some} nonparametric structural assumption which guarantees that the function $m$ is unique under this. There are plenty of these structural assumptions. In fact, one general class of functions comes from the theory of optimal transport and can be used for the identification of Hedonic models with multivariate heterogeneity and endogenous characteristics, our main application. We call those production functions \emph{determinable}. 

\begin{definition}\label{uniquelyidentifiable}
A measurable production function $m:\mathcal{E}\to \mathcal{Y}_x$ is \emph{determinable} if the class of functional form assumptions on $m$ intersected with the class of measure preserving isomorphisms between $P_\varepsilon$ and $P_{Y|X=x}$ contains a unique element for all $x\in\mathcal{X}$.
\end{definition}
Requiring $m(x,\varepsilon)$ to be strictly continuous and increasing in $\varepsilon$ is a functional form restriction which makes it determinable, if we in addition rule out the existence of strictly increasing and continuous transformations $g\circ m$ of $m$, because for given $F_\varepsilon$ $m$ and $g\circ m$ are observationally equivalent \citep[Lemma 1]{matzkin2003nonparametric}. Let us give some other examples of determinable measure-preserving isomorphisms. The following example is in the univariate case.
\begin{example}[A matching example on the real line with concave cost functions]\label{mccannexample}
\citet{mccann1999exact} introduces a matching model on the real line $\mathbb{R}$ between a distribution $P_\varepsilon$ of suppliers (e.g.~coal mines) of some product and a distribution $P_Y$ of demanders (e.g.~factories) of this product. The problem is this model consists of minimizing the transport cost between coal mines and factories. The cost of transportation is modeled as $c(y-e)$, $y\in\mathcal{Y}$, $e\in\mathcal{E}$ for strictly concave $c$. McCann argues that a concave cost function of the transport distance provides a reasonable model for applications in which shipping occurs along a single route, because the resulting shipping routes display economies of scale. He goes on to show that solving this optimal transport problem under the specified concave cost function has a unique measure preserving isomorphism as a solution $Y = m(\varepsilon)$, which is not monotone. Assuming that $m$ is the solution of the above optimal matching for concave costs is a functional form restriction which makes $m$ determinable. 
\end{example}
Since $m$ is not strictly increasing, it cannot be identified by current approaches in the literature.
As another example, note that multivariate monotonicity also leads to determinability in higher dimensions. 
\begin{example}[Multivariate monotonicity] For absolutely continuous probability measures $P_\varepsilon$ and $P_{Y|X=x}$ with finite second order moments there exists a unique measure preserving isomorphism $m(x,\cdot)$ transporting $P_\varepsilon$ onto $P_{Y|X=x}$ for all $x$, which takes the form of a gradient of a convex function by a famous result in \citet{brenier1991polar}, i.e.~$m(x,\cdot)\coloneqq \nabla\varphi_x(\cdot)$ for convex $\varphi_x$. \citet{mccann1995existence} generalized Brenier's theorem to show that this result holds even if the probability measures do not have finite second order moments. We have included this result in the appendix (Theorem \ref{mccann}) for the sake of completeness. Therefore, assuming that $m(x,\cdot)$ is the gradient of a convex function for all $x\in\mathcal{X}$ is a functional form restriction which makes $m$ determinable.
\end{example}

Gradients of convex functions are the most natural generalization of increasing and continuous functions. In particular, if $T=\nabla\varphi$ for some convex $\varphi$, then it is monotone in the following sense \citep[p.~53]{villani2003topics}:
\[\langle T(x)-T(z),x-z\rangle\geq 0,\] where $\langle\cdot,\cdot\rangle$ denotes the inner product on $\mathbb{R}^d$. Here it is easy to see that if $x>z$ in the partial ordering induced by the positive cone on $\mathbb{R}^d$, then this definition implies that $T(x)\geq T(z)$ or that $T(x)$ and $T(z)$ are not comparable. This monotonicity property has been exploited recently by \citet{carlier2016vector} who use it to generalize the concept of univariate quantiles. 

In general, functional form restrictions will most likely come from solving functional equations in economic theory, like a minimization problem for demand functions, for example. If these functional equations admit a unique invertible solution, then they induce a functional form restriction which make their solution determinable.

\begin{example}[Demand function]
Following \citet{matzkin2007heterogeneous}, we denote by $V(y,x,\varepsilon)$ the indirect utility function of a consumer over bundles of goods $y$, where $x$ are observable and $\varepsilon$ are unobservable characteristics. Then a demand function can be obtained by
\[d(p,I,x,\varepsilon)\coloneqq \argmin_{y}\{V(y,x,\varepsilon) : p\cdot y\leq I\},\] for $p$ the price vector and $I$ the initial endowment. The standard assumption is then that $d$ is the unique solution and invertible in $\varepsilon$ (e.g.~\citeauthor{berry2014identification} \citeyear{berry2014identification}, p.~1757), which makes $d$ determinable.
\end{example}

Note that in order to identify a determinable $m$ in practice, one needs to make a normalization assumption, usually on the unobservable distribution $F_\varepsilon$, to guarantee that there is only a unique set of $(m,F_\varepsilon)$ which can generate the distribution $F_{Y|X=x}$ for all $x\in\mathcal{X}$. This should be intuitively clear as one in principle needs to identify two things, the unobservable $F_\varepsilon$ as well as the corresponding production function $m$.

\subsection{Main assumptions for the theoretical main result}\label{assumptionssubsection}
We can now lay out the main assumptions for the theoretical identification result.
A convenient property of our approach is that one can use the assumptions from \citet{torgovitsky2015identification} for the first stage $X=h(Z,U)$, i.e.~requiring that $h$ can be written element-wise as $h(Z,U) = [h_1(Z,U_1),\ldots,h_k(Z,U_k)]'$ and require the univariate functions $h_i$ to be strictly increasing and continuous in each univariate $U_i$, which makes our approach a direct generalization. This, however, requires the strong assumption of a compact and rectangular support for $\mathcal{X}|Z=z$ for all $z\in\mathcal{Z}$ (\citeauthor{torgovitsky2015supplement} \citeyear{torgovitsky2015supplement} and \citeauthor{d2015identification} \citeyear{d2015identification}). 

We propose a general multivariate approach which allows for weaker assumptions on the support of $F_{X|Z=z}$ and $h$, but requires an additional normalization assumption on the first stage, which fixes the distribution of $U$. We assume that the distribution of $U$ is known.\footnote{An alternative approach would be the general control variable approach, which we introduce in \citet{gunsilius2017identification}.} With this assumption we can make nonparametric functional form assumptions on $h$, like requiring it to be the gradient of a convex function itself.\footnote{It will be important to work with the gradients of convex functions for our main identification result. In particular, the whole identification result rests on an apparently new insight into properties of Brenier's theorem about gradients of convex functions mapping between two multivariate probability distributions, which we prove in Lemma \ref{brenierismetricprojection} in the appendix.}

\begin{assumption}\label{mpiso}
The production function $h(z,u)$ in first stage of model \ref{triangular} takes the form of the gradient of a convex function between $U$ and $X$ for all $z\in\mathcal{Z}$. Moreover, $U$ has an absolutely continuous distribution $F_U$.
\end{assumption}
The following assumption is the normalization assumption on $U$.
\begin{assumption}\label{normalizationass}
There is a known $\bar{z}\in\mathcal{Z}$ such that $h(\bar{z},u)=u$ for all $u\in\mathcal{U}$.
\end{assumption}
Assumption \ref{normalizationass} is a standard assumption made in the literature and was first proposed in \citet{matzkin2003nonparametric}. It requires $h$ to be the identity mapping for a known $\bar{z}$ between $X$ and $U$ and hence fixes the distribution of $U$. This assumption has also been used in other areas, most notably the measurement error literature, where \citet{hu2008instrumental} require a known functional which fixes the unobservables in a certain rotation. Assumption \ref{normalizationass} allows us to relax the strong assumption of a rectangular and compact support on the observable $F_{X|Z=z}$:
\begin{assumption}\label{abscont}
The distributions $F_{X|Z=z_i}$ are absolutely continuous with convex support $\mathcal{X}_{z_i}\subseteq\mathbb{R}^k$ for all $z_i\in\mathcal{Z}$ and are strictly increasing.
\end{assumption} 
We allow for $Z$ to be a discrete and even binary instrument.\footnote{We have included the statement and the proof of the theorem for continuous $Z$ in the appendix, because it is not relevant for the identification of multi-market Hedonic models.} In the following we focus on the binary case, i.e.~$\mathcal{Z}=z,z'$, because in the proof of the main result we assume $h$ to be the identity for one $z$ and the gradient of a convex function for the other $z'$. Also, for the identification of Hedonic models, two markets are sufficient by definition, and adding more markets would not change the result in any way. 

\begin{assumption}\label{instrumentass}
$Z$ is a valid instrument for $X$, i.e.~(i) it generates exogenous variation in $X$ such that $F_{X|Z=z}(x)\neq F_{X|Z=z'}(x)$ for at least one $x\in\mathcal{X}$ and (ii) is independent of $\varepsilon$ and $U$, denoted by $(\varepsilon,U)\independent Z$.
\end{assumption}
Note that we do \emph{not} require the dimensionality of $Z$ to be at least the dimensionality of $X$. In fact, and this is analogous to the result in \citet{torgovitsky2015supplement}, $Z$ can be lower- and even one-dimensional for multivariate $X$ and the approach still works. This is important for the application to Hedonic models, because the idea is to consider each market to be the realization of a one-dimensional instrument. Intuitively, we can allow for lower-dimensional $Z$ because we make use of the nonlinear and nonparametric structure, i.e.~taking into account the information of all higher order and not just first-order moments; this will become apparent momentarily when we introduce Assumption \eqref{supportass}.

The ultimate goal for identification is to point-identify the (multivariate) second stage production function $m$.
\begin{assumption}\label{mpiso2}
$m(x,\varepsilon)$ is a determinable measure preserving isomorphism between $Y$ and $\varepsilon$ and is continuous in $X$ for all $x$. Moreover, $P_\varepsilon$ is known.\footnote{In the univariate case following \citet{matzkin2003nonparametric}, $P_\varepsilon$ is usually normalized to be the uniform distribution. Since we work with general and multivariate distributions, we only assume it is known. Alternatively, instead of assuming $P_\varepsilon$ is known, one could assume that $m(\bar{x},e)=e$ at some known $\bar{x}\in\mathcal{X}$ for all $e\in\mathcal{E}$, i.e.~assuming that it is the identity map for $\bar{x}$, just as for the first stage. Then $P_\varepsilon$ is also fixed, because by measure preservation of $m$ we have $P_{Y|X=\bar{x}}(E) = P_{\varepsilon}(m^{-1}(\bar{x},E)) = P_\varepsilon(E)$ for every Borel set $E$ in $\mathscr{B}_{\mathbb{R}^d}$.}
\end{assumption}
It turns out, however, that assuming continuity of $m(x,\varepsilon)$ in $x$ is a rather strong assumption. In particular, we will not be able to guarantee this in our efforts to identify the multi-market Hedonic model in the next section. Therefore, we need to make a weaker assumption. This requires the definition of convergence in measure, see for instance \citet[Definition 2.2.2]{bogachev2007measure1}.
\begin{definition}
Suppose we are given a measure space $(\mathcal{X},\mathscr{A}_X)$ with a probability measure $P$ and a sequence of $P$-measurable functions $\{f_n\}_{n\in\mathbb{N}}$. Then the sequence $\{f_n\}_{n\in\mathbb{N}}$ is said to converge in measure to a $P$-measurable function $f$ if for every $c>0$ one has
\[\lim_{n\to\infty} P(x:|f_n(x)-f(x)|\geq c)=0.\]
\end{definition}
Based on this, we only require $m(x,\varepsilon)$ to be continuous in $x$ in the following weaker sense.
\begin{definition}\label{continprop}
If for every sequence $\{x_n\}_{n\in\mathbb{N}}\in\mathcal{X}$ which converges to some $x\in\mathcal{X}$ the corresponding sequence $m(x_n,\varepsilon)$ converges in measure to $m(x,\varepsilon)$, we say that $m$ is continuous in measure.
\end{definition}
The weaker assumption we need to require for $m$ is hence that it is continuous in measure based on Definition \ref{continprop}.
\begin{assumptions}{6}{'}\label{mpiso2prime}
$m(x,\varepsilon)$ is a determinable measure preserving isomorphism between $Y$ and $\varepsilon$ and is continuous in measure in $x$. Moreover, $P_\varepsilon$ is known. 
\end{assumptions}

We also have to make a large support assumption on $\varepsilon$. This is the same assumption that \citet{torgovitsky2015identification} had to impose.
\begin{assumption}\label{supportepsilon}
The support $\mathcal{E}$ is convex and independent of $X=x$ and $Z=z$, i.e.~$\mathcal{E}_{x,z}$ coincides with $\mathcal{E}$ for all $(x,z)\in\mathcal{X}\times\mathcal{Z}$. 
\end{assumption}

Now, in order to identify the model with discrete or binary instruments, the main assumption needs to be made on the distribution functions $F_{X|Z=z_i}$, $i=1,\ldots,k_z$. It is the analogous assumption to the ones made in \citet{torgovitsky2015identification}, who requires all univariate distribution functions to intersect in at least one point. In the following we consider binary $Z$ with realizations $z$ and $z'$; furthermore, for each pair $z,z'\in\mathcal{Z}$ we denote as $\mathcal{I}(z,z')\subseteq \mathcal{X}_z\cup\mathcal{X}_{z'}$ the set where $F_{X|Z=z}$ and $F_{X|Z=z'}$ intersect. That is,
\[\mathcal{I}(z,z')\coloneqq\{x\in\mathcal{X}_z\cup\mathcal{X}_{z'}: 0<F_{X|Z=z}(x)=F_{X|Z=z'}(x)<1\}.\] Moreover, for every point $x_0\in\mathcal{X}_z=\mathcal{X}_{z'}$, we let $I_z(x_0)$ denote the \emph{isoquant} or \emph{level set} of $F_{X|Z=z}$ at $x_0$, which is the set of all $x\in\mathcal{X}_z$ which have the same probability as $x_0$ under $F_{X|Z=z}$, formally:
\[I_z(x_0)\coloneqq\{x\in\mathcal{X}_z:F_{X|Z=z}(x)=F_{X|Z=z}(x_0)\}.\]
Since we assumed that $F_{X|Z=z}$ is absolutely continuous with convex support and strictly increasing in a multivariate sense, we can regard the distribution functions $F_{X|Z=z}$ and $F_{X|Z=z'}$ as utility functions, in which case the isoquants $I_z(\cdot)$ and $I_{z'}(\cdot)$ can be interpreted as the indifference curves (or in higher dimensions: indifference manifolds) of $F_{X|Z=z}$ and $F_{X|Z=z'}$, respectively. This analogy is the key in proving the result as it allows us to work with the indifference curves instead of the distribution functions. 

Now for stating the main assumption which gives us identification, we need to introduce the concept of transversal intersection of manifolds. The following definition is adapted from \cite{milnor1997topology}.
\begin{definition}
Two submanifolds $N$ and $N'$ of an ambient manifold $M$ intersect transversally if for each $x\in N\cap N'$ their tangent spaces at $x$, denoted by $T_x N$ and $T_x N'$, together generate the tangent space $T_x M$ in the sense that $T_x N+T_xN'=T_xM$.
\end{definition}
Transversal intersection of indifference curves of different utility functions is a standard assumption made in economic theory \citep{mas1989theory} and is very weak since it is a \emph{generic} property in the sense that basically all indifference curves between different preferences intersect transversally by a result from Ren\'e Thom (see \citeauthor{ekeland2004identification} \citeyear{ekeland2004identification} for a discussion of generic properties). 
We are now in the position to state the main assumption on the distributions $F_{X|Z=z}$.

\begin{assumption}\label{supportass}
Let $Z$ be binary with $\mathcal{Z}=\{z,z'\}$. Then the following properties of $F_{X|Z=z}$ and $F_{X|Z=z'}$ hold:
\begin{enumerate}
\item $\mathcal{X}_{z}=\mathcal{X}_{z'}$.
\item The epigraphs\footnote{The epigraph of a real valued function $f:X\to\mathbb{R}$ for the level $\alpha\in\mathbb{R}$ is defined by $\text{epi}(f;\alpha)\equiv\{(x,\alpha)\in\mathcal{X}\times\mathbb{R}:\alpha\geq f(x)\}$, see \citet{aliprantis2006infinite}.} of $F_{X|Z=z}$ and $F_{X|Z=z'}$, denoted by $\text{epi}(F_{X|Z=z};\alpha)$ and $\text{epi}(F_{X|Z=z'};\alpha)$ for $\alpha\in[0,1]$, are convex sets. The set of all points where the isoquants meet, $\mathcal{I}(z,z')$, consists of at least one connected manifold $\mathcal{M}(z,z')$. At all points $x_0\in\mathcal{I}(z,z')$ the isoquants either intersect transversally or coincide in a neighborhood $\mathcal{N}(x_0)$ around $x_0$. 
\item The manifold $\mathcal{M}(z,z')$ is such that

\noindent (i) for each $x\in\mathcal{X}_{z}=\mathcal{X}_{z'}$ there is an $m\in\mathcal{M}(z,z')$ with $0<F_{X|Z=z}(m)=F_{X|Z=z'}(m)<1$ which either dominates or is dominated by $x$.

\noindent (ii) All points $x\in\mathcal{X}_{z}$ with $F_{X|Z=z}(x)=0$ lie on one side of the manifold and all points where $F_{X|Z=z}(x)=1$ lie on the other, and analogously for $F_{X|Z=z'}$; ``lying on one side of the manifold'' means that there are no two points $x_1,x_2\in\mathcal{X}_{z}$ with $F_{X|Z=z}(x_1)=F_{X|Z=z}(x_2)=0$ (respectively: $F_{X|Z=z'}(x_1)=F_{X|Z=z'}(x_2)=1$) such that the line $(1-t)x_1+tx_2$ for $t\in[0,1)$ intersects $\mathcal{M}(z,z')$.
\end{enumerate}
\end{assumption}
Part 1 of Assumption \ref{supportass} is restrictive as it requires that the supports of all conditional distribution functions coincide. This assumption is slightly stronger than the assumption in the univariate case of \cite{torgovitsky2015identification} as the distribution functions there only need to intersect but the supports need not coincide. On the other hand, we allow for the supports to be convex and unbounded, a much weaker assumption. Parts 2 and 3 of Assumption \ref{supportass} appear to be high level, but are actually rather natural, weak, and easy to check in practice. 

To see that Assumption \ref{supportass} is a reasonable assumption to make in practice, consider Figure \ref{cdfplot} as an example, where we display the intersection of a bivariate $t$ distribution $F_{X|Z=z}$ with density function
\begin{align*}
&f_{X|Z=z'} = (v\pi)^{-1}|\Sigma|^{-\tfrac{1}{2}}\frac{\Gamma(\tfrac{1/2}v+1)}{\Gamma(\tfrac{1}{2}v)}\left(1+\frac{x^{T}\Sigma^{-1}x}{v}\right)^{-\tfrac{1}{2}v-1}\thickspace\medspace\text{for}\thickspace\medspace v=2\thickspace\medspace\text{and}\thickspace\medspace \Sigma=\begin{pmatrix} 2&0.8\\0.8&0.5\end{pmatrix}\\
&\text{with}\thickspace\medspace\Gamma(z) = \int_0^\infty x^{z-1}\exp(-x)dx
\end{align*}
and a bivariate Normal distribution $F_{X|Z=z'}$ with density function
\[f_{X|Z=z} = (2\pi)^{-1}|\Sigma'|^{-\tfrac{1}{2}}\exp(-\tfrac{1}{2}(x-\mu)^{T}(\Sigma')^{-1}(x-\mu))\quad\text{for}\thickspace\medspace \mu=(0,0)'\thickspace\medspace\text{and}\thickspace\medspace \Sigma'=\begin{pmatrix} 1&0.8\\0.8&4\end{pmatrix}.\]
\begin{figure}[ht]
\includegraphics[width=15cm,height=8cm]{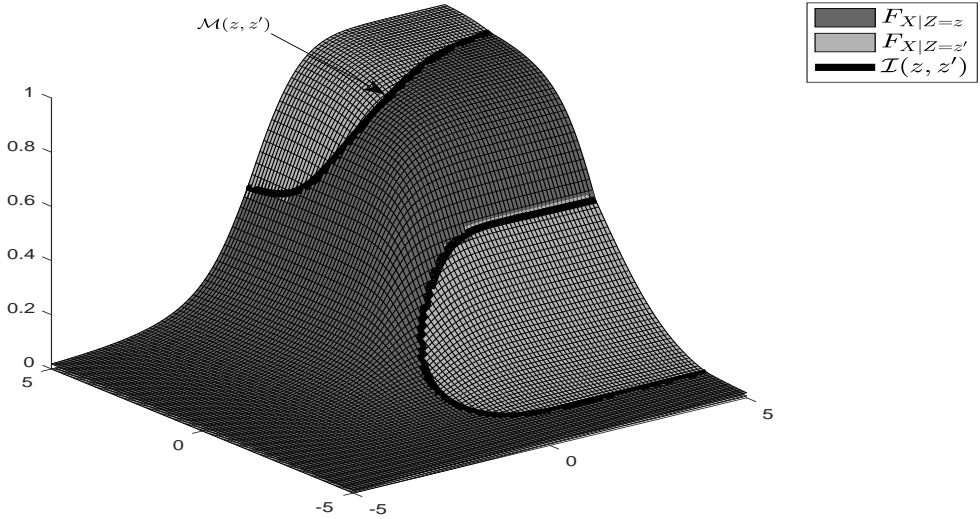}
\caption{Example of a bivariate Normal distribution and a bivariate t distribution}\label{cdfplot}
\end{figure}
In this case, the set where all isoquants intersect, $\mathcal{I}(z,z')$, consist of two separate manifolds. The one on the bottom does not satisfy part 4 of Assumption \ref{supportass}, but the one on top, $\mathcal{M}(z,z')$, does. This is all we need since we only need one manifold to satisfy Assumption \ref{supportass}. To check part 3 of Assumption \ref{supportass} we need to look at the contour plot. Figure \ref{contourplot} shows that the respective isoquants either intersect transversally (in the interior of the graph) or converge towards one another (at the boundaries) so that they coincide there. Assumption \ref{supportass} is hence satisfied which would guarantee point-identification of $m$ in a nonseparable triangular model where $F_{X|Z=z}$ and $F_{X|Z=z'}$ are those two distributions. 

We want to mention in this respect that even though the form of $I(z,z')$ depends on $\Sigma$ and $\Sigma'$, identification holds for all combinations of $\Sigma$ and $\Sigma'$ and degrees of Freedom $v$ simply by the fact that both $t$- and Normal distribution have infinite support. That is, in all cases there is a manifold $\mathcal{M}(z,z')\subseteq I(z,z')$ with the required properties simply because the CDFs need to intersect at some point in the infinite support. Note that we have chosen an example where the variance of $F_{X|Z=z}$ does not exist since $v=2$. Our approach still works in this case since the gradient of convex functions mapping one distribution to the other exists (see Theorem \ref{mccann} in the appendix). The same holds if we choose two Normal distributions. This makes us confident that Assumption \ref{supportass} is satisfied in many important practical applications, not just in two but also higher dimensions. This is why our result holds more generally than the results of \citet{torgovitsky2015identification} and \citet{d2015identification}, as the authors there need to assume a compact support in their multivariate settings. That said, there are certainly cases which do not satisfy Assumption \ref{supportass}. In particular, we cannot allow for the fact that one distribution first-order stochastically dominates the other distribution in a multivariate sense, but this requirement is the same as in \cite{torgovitsky2015identification} or \cite{d2015identification}, only in the multivariate case.

\begin{figure}[htb]
\centering
\includegraphics[width=12cm,height=8cm]{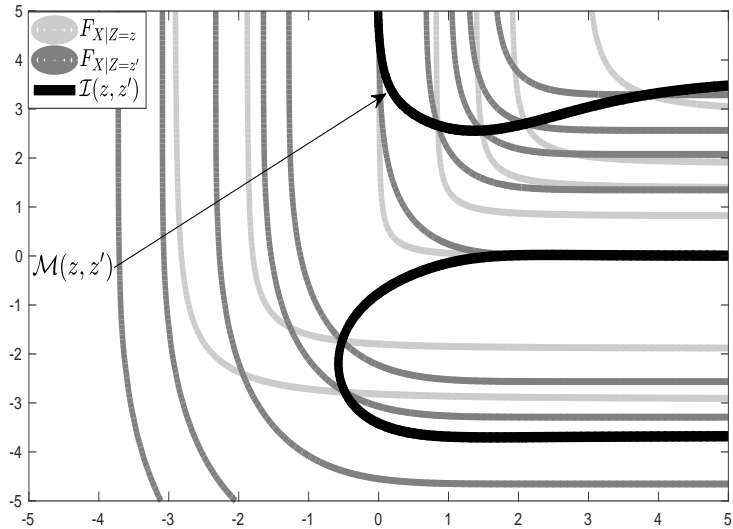}
\caption{Contour plot corresponding to Figure \ref{cdfplot}}\label{contourplot}
\end{figure}
In the univariate case, i.e.~when $X,U\in\mathbb{R}$, the manifold $\mathcal{M}(z,z')$ reduces to a point in the special case where $\mathcal{X}_z$ and $\mathcal{X}_{z'}$ coincide. In this case we can simply use Torgovitsky's assumption on the first stage which makes our approach a direct generalization of \citet{torgovitsky2015identification}. 

To get a better idea of when Assumption \ref{supportass} holds in general, consider Figure \ref{sassfigure} and suppose that the sheet of paper represents $\mathbb{R}^2$ with the standard partial order $x\geq y$ if and only if $x_1\geq y_1$ and $x_2\geq y_2$. Depicted there are three different scenarios for the supports $\mathcal{X}_z$ and $\mathcal{X}_{z'}$ as well as the manifold $\mathcal{M}(z,z')$ in these supports. These pictures are schematic versions of Figure \ref{contourplot} in that they only depict the supports $\mathcal{X}_z$ and $\mathcal{X}_{z'}$ and the respective manifolds $\mathcal{M}(z,z')$, but not the contour plots. Consider the picture on the left first. Any one of those three depicted manifolds in this example satisfies Assumption \ref{supportass} as all three are of dimension 1 and connected; moreover, for each point $x\in\mathcal{X}_z=\mathcal{X}_{z'}$, there is a point in every manifold which either dominates or is dominated by $x$.\footnote{Recall that a point $x$ dominates $x'$, i.e.~$x>x'$, if it lies to the ``north-east'' of $x'$.} Also all points in the support with $F_{X|Z=z}(x)=0$ lie on the same side of the manifolds, and analogously for $F_{X|Z=z'}$. The example in the center violates Assumption \ref{supportass} since $x_1$ with $F_{X|Z=z}(x_1)=0$ and $x_2$ with $F_{X|Z=z}(x_2)=0$ lie on opposite sides of $\mathcal{M}(z,z')$.\footnote{Note that $F_{X|Z=z}(x_1)=0=F_{X|Z=z'}(x_2)$, because the rectangles $(-\infty,x_1]$ and $(-\infty,x_2]$ do not intersect the support, so that $F_{X|Z=z}(x_1) = P_{X|Z=z}((-\infty,x_1])=0=P_{X|Z=z}((-\infty,x_2])=F_{X|Z=z}(x_2)$.} Finally, the example on the right violates Assumption \ref{supportass} even though there are \emph{two} connected manifolds which \emph{together} are such that each point $x$ either dominates or is dominated by some $m$ in one of the manifolds. The problem here is that there is not \emph{one} manifold alone for which this holds. In fact, there is no point in the top left manifold which either dominates or is dominated by $x_4$, i.e.~lies to the north-east or south-west; similarly, there is no point on the bottom right manifold which either dominates or is dominated by $x_3$. In addition, both manifolds violate the requirement that all points with $F_{X|Z=z}(x)=0$ and $F_{X|Z=z'}(x)=0$ lie on one side.

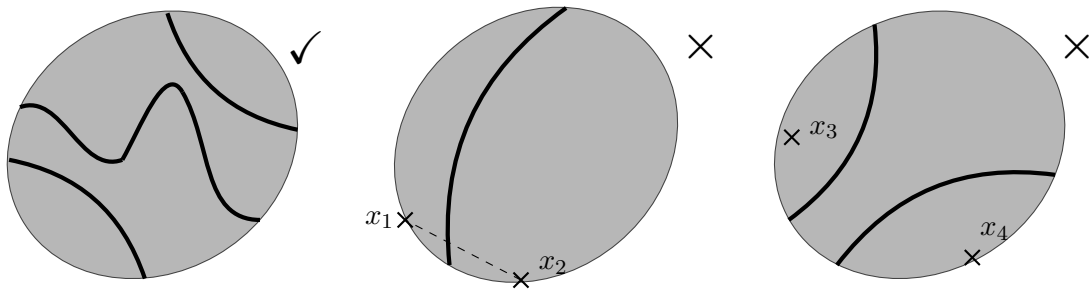
\begin{figure}[htp]
\centering
\begin{tikzpicture}
\def\firstcircle{(2.4,2.2) circle [x radius=2cm, y radius=1.7cm, rotate=30]};
\def\secondcircle{(12.6,1.8) circle [x radius=2cm, y radius=1.7cm, rotate=50]};
\draw[fill=gray!80, opacity=0.7] \firstcircle;
\draw [ultra thick] (2.6,3.95) to[bend right] (4.33,2.4);
\draw [ultra thick] (0.64,2.7) to [out=20,in=200] (2,2);
\draw [ultra thick] (2,2) to [out=60,in=120] (2.8,2.9);
\draw [ultra thick] (2.8,2.9) to [out=300,in=180] (3.83,1.2);
\draw [ultra thick] (0.495,2) to[bend left] (2.3,0.42);
\draw node[right] at(5,5.25) {$\mathcal{X}_z=$};
\draw[fill=gray!80, opacity=0.7] (6.2,5) rectangle (6.7,5.5);
\draw node[right] at(6.7,5.25) {$=\mathcal{X}_{z'}$};
\draw[ultra thick] (8.2,5.25) -- (8.7,5.25);
\draw node[right] at(8.7,5.25) {$=\mathcal{I}(z,z')$};
\def\thirdcircle{(7.5,2.2) circle [x radius=2cm, y radius=1.7cm, rotate=40]};
\def\helpcircle{(9.61,2.1) circle [x radius=2cm, y radius=1.7cm, rotate=61]};
\draw[fill=gray!80, opacity=0.7] \thirdcircle;
\draw [ultra thick] (6.35,0.6) to[bend left] (7.9,4.02);
\draw [thick] (7.2,0.3) to (7.4,0.5);
\draw [thick] (7.2,0.5) to (7.4,0.3);
\draw node[right] at (7.4,0.6) {$x_2$};
\draw node[left] at (5.75,1.2) {$x_1$};
\draw [thick] (5.66,1.1) to (5.86,1.3);
\draw [thick] (5.66,1.3) to (5.86,1.1);
\draw[dashed] (7.3,0.4) to (5.76, 1.2);
\def\fourthcircle{(12.6,2.2) circle  [x radius=2cm, y radius=1.7cm, rotate=30]};
\draw[fill=gray!80, opacity=0.7] \fourthcircle;
\draw [ultra thick] (10.86,1.2) to[bend right] (12,3.8);
\draw [ultra thick] (14.4,1.8) to[bend right] (11.5,0.6);
\draw node[above] at (13.6,0.8) {$x_4$};
\draw [thick] (13.2,0.6) to (13.4,0.8);
\draw [thick] (13.2,0.8) to (13.4,0.6);
\draw [thick] (10.8,2.2) to (11,2.4);
\draw [thick] (10.8,2.4) to (11,2.2);
\draw node[right] at (11,2.4) {$x_3$};
\draw node[right] at(4.05,3.5) {\LARGE{$\checkmark$}};
\draw node[right] at(9.3,3.5) {\LARGE{\textbf{\texttimes}}};
\draw node[right] at(14.3,3.5) {\LARGE{\textbf{\texttimes}}};
\end{tikzpicture}
\caption{Examples (not) satisfying Assumption \ref{supportass}}
\label{sassfigure}
\end{figure}

Assumption \ref{supportass} is also weaker than the currently existing assumptions in another respect. In particular, we only require one manifold $\mathcal{M}(z,z')$ for binary $Z$.  In contrast, \cite{torgovitsky2015supplement}  requires that for every endogenous variable $X_i$, $F_{X_i|Z=z}$ and $F_{X_i|Z=z'}$ have to intersect in at least one point, and this for every $i=1,\ldots,d_x$. So in a $k$-variate case, this would actually require $k$-linear curves, each orthogonal to one of the $k$ dimensions, instead of one general curve from Assumption \ref{supportass}.

All Assumptions are rather weak and can even be checked by estimating the respective distribution functions $\hat{F}_{X|Z=z}$ and $\hat{F}_{X|Z=z'}$ and examining wether their intersection satisfy Assumption \ref{supportass}. Note, however, that, analogous to \citet{torgovitsky2015identification}, this support assumption excludes linear relationships like $X=\beta Z+U$ with $\mathcal{U} = \mathbb{R}^k$ for $\mathcal{Z}=\{0,1\}$ with bounded supports, because in those relationships the two conditional distribution functions would simply be a shift of each other and would not intersect for $\beta\neq 0$.

Still, by our above reasoning and the fact that the assumption is satisfied by Normal distributions, we are convinced that Assumption \ref{supportass} holds in many practical settings. 

\subsection{The theoretical main result and intuition of the proof}\label{mainresultsubsection}
Under the above mentioned assumptions, we can now state the main theoretical result of this article.
\begin{theorem}\label{maintheorem2}
Let Assumptions \ref{dimensionsass} -- \ref{supportass} hold and let $Z$ be discrete with at least two points in its support. Then $m(x,\varepsilon)$ is identified for almost every $x\in\mathcal{X}$ and all $\varepsilon\in\mathcal{E}$ in model \eqref{triangular}. The identified set \[\mathbb{I}\coloneqq \{m\in\mathcal{H}(Y_{xz},\varepsilon_{xz}):(m^{-1}(X,Y),U)\independent Z\}\] hence contains an $X$-almost everywhere unique element. Here, $\mathcal{H}(Y_{xz},\varepsilon_{xz})$ denotes the set of all measure preserving isomorphisms between $P_{Y|X,Z}$ and $P_{\varepsilon|X,Z}$ satisfying Assumption \ref{mpiso2}.
If Assumption \ref{mpiso2prime} holds in place of \ref{mpiso2}, the analogous result holds, but $m$ is then only identified for almost every $\varepsilon\in\mathcal{E}_{xz}$ instead of all $\varepsilon$.
\end{theorem}
The importance of Assumption \ref{mpiso2prime} is that it basically does not weaken the result (identification for almost every $\varepsilon$ compared to identification for every $\varepsilon$) compared to Assumption \ref{mpiso2}, while being a rather substantial weakening of Assumption \ref{mpiso2}. In particular, we can guarantee continuity in probability of $m$ but not full continuity in the next section for identification of Hedonic models. 

We have relegated the proof of this result to the appendix, but let us give an outline of the idea. Intuitively, the problem of identification results from the fact that $m$ is the map between $P_\varepsilon$ and $P_{Y|X}$ for exogenous $X$. If $X$ were actually exogenous, we would not need a first stage relationship, because in this case the observable distribution $F_{Y|X}$ is exactly the distribution corresponding to $m$ and we could simply use the observable distribution and a normalization of $F_\varepsilon$ to identify $m$. This is the underlying idea for identification of single market Hedonic models with exogenous characteristics. 

Since $X$ is endogenous, however, the observable distribution $F_{Y|X}$ is not the right distribution for identifying $m$. We therefore need an instrument $Z$ which has a nonzero influence on $X$ and is independent of $\varepsilon$. Then for a binary (or discrete) $Z$ with values $z$ and $z'$ and the first stage relationships $X=h(z,U)$ and $X=h(z',U)$, we can use $U$ with known distribution $F_U$ as a control variable in the sense of \citet{imbens2009identification}, only in a multivariate setting and not requiring $U$ to have a uniform distribution. In fact, note that by the assumption that $h$ is the gradient of a convex function mapping $U$ to $X$ for $z'$ and the identity map between $X$ and $U$ for $z$ and the fact that both $P_{X|Z}$ and $P_U$ are absolutely continuous, $h$ establishes a bijective relation between $X$ and $U$ for $z$ as well as $z'$. That is, for each $u\in\mathcal{U}$ there are two $x,x'\in\mathcal{X}$, possibly coinciding, corresponding to it: $u=h^{-1}(x,z)$ and $u=h^{-1}(x',z')$. If we fix $Z=z$, then the relation is bijective. In the other direction, for every $x\in\mathcal{X}$ there are two $u,u'\in\mathcal{U}$: $x=h(z,u)$ and $x=h(z',u')$. The second crucial ingredient is the measure preservation of $h$. In fact, for every Borel set $E_u\in\mathscr{B}_{\mathbb{R}^k}$ we have $P_U(E_u) = P_{X|Z=z}(h^{-1}(E_u,z))$ and similarly for $z'$, so that the distributions do not change if we condition on $U$ or $X$. 

Therefore, for fixed $Z$, conditioning on $U$ is the same as conditioning on $X$. Now the crucial assumption guaranteeing that $Z$ is independent of $\varepsilon$ and $U$ allows us to use the fact that for each $u$ there are two $x,x'$, depending on which realization of $Z$ we use for the map $h$.  Then the idea---for all $x\in\mathcal{X}$---is to construct a sequence from $x$ to $u$ via $h^{-1}(\cdot,z)$, and then change to $x'$ via $x'=h(z',u)$, then change to $u'$, and so forth. The key here is that for this sequence starting with any $x$ the distributions $F_{Y|X}$ and $F_\varepsilon$ of the second stage \emph{do not change} because of the independence of $Z$ and $\varepsilon$. So this sequence induces an exogenous change in $X$ by changing $Z$ which does not affect the distribution of $F_\varepsilon$. By assumption \ref{supportass}, this sequence must converge and cannot go on forever, because at some point it must be that $x=x'$. This holds for every starting point $x$, so that we can in principle identify $m$ by exogenously varying $Z$. 

In practice, we do not observe $F_{Y|X}$ for exogenous $X$, even using the instrument $Z$. Theorem \ref{maintheorem2} hence only shows that $m$ is \emph{identifiable}, and gives us the identification set, but not a constructive way to obtain $m$. This reasoning is perfectly analogous to the result in \citet{torgovitsky2015identification} and also \citet{d2015identification}. Note again that the dimension of $Z$ can be smaller and even one-dimensional for this, as long as $Z$ is a valid instrument for each variable in the vector $X$. This works since we require nonlinearities of $m$ by way of Assumption \ref{supportass} and hence implicitly take into account the information of all higher order moments as mentioned.

To be slightly more formal: the basic idea is to prove uniqueness of $m$. So it is natural to assume that there are $m$ and $m^*$ as well as corresponding $\varepsilon$ and $\varepsilon^*$ satisfying the assumptions and $Y=m(X,\varepsilon)$ as well as $Y=m^*(X,\varepsilon^*)$. Identifiability can then be proved if $m=m^*$. This is done by showing that the isomorphism \[q(x,z,\cdot)=q(x,\cdot)=m^{-1}(x,m^*(x,\cdot))\] is actually the identity, i.e. that $q(x,e)=e$ for all $x\in\mathcal{X}_z\cup\mathcal{X}_{z'}$ and $e\in\mathcal{E}$, which would imply $m=m^*$ for every $x$. To show that $q(x,e)$ is the identity with respect to $\varepsilon$, it is actually sufficient to show that it is not a function of $x$ by the fact that $F_\varepsilon$ is known and $m$ is determinable. \label{explanation} In fact, as $m(x,e)$ is determinable between $F_\varepsilon$ and $F_{Y|X=x}$ for each $x\in\mathcal{X}$, there can be no other $m(x,\cdot)$ of this functional form by definition for each $x$. Therefore, if $q(x,e)$ is only a function of $e$, say $f(e)$, this means that the functional form of $m^{-1}\circ m^*$ does not change with $x$ so that both $m$ and $m^*$ have the same functional form. But since $m$ is determinable, it must be that $m=m^*$. This is the same reasoning as in \cite{torgovitsky2015identification}, only put in a more general framework. In order to achieve this, Assumption \ref{instrumentass} is crucial, as it guarantees that $P_{\varepsilon|X,Z}=P_{\varepsilon|X}$ and analogously for $\varepsilon^*$. 

\begin{figure}[h!t]
\centering
\begin{tikzpicture}
\draw node[below] at (1.5,0) {$P_{\varepsilon^*|X=x,Z=z}=P_{\varepsilon^*|U=h^{-1}(x,z),Z=z}$};
\draw node[above] at (-0.5,2.5) {$P_{Y|X=x,Z=z}$};
\draw node[above] at (7.2,2.4) {$P_{\varepsilon|X=x,Z=z}=P_{\varepsilon|U=h^{-1}(x,z),Z=z}$};
\draw node[left] at (0,1.25) {$m^*(x,\cdot)$};
\draw node[above] at (2.75,2.75) {$m(x,\cdot)$};
\draw node[right] at (2.4, 0.9) {$q(x,z,\cdot)=q(h^{-1}(x,z),\cdot)$};
\draw[->,thick] (0,0.15) to (0,2.5);
\draw[<-,thick] (0.85,2.75) to (4.25,2.75);
\draw[->,thick] (0.6,-0.05) to (4.5,2.5);
\end{tikzpicture}
\caption{Underlying isomorphism structure}
\label{isostructure}
\end{figure}
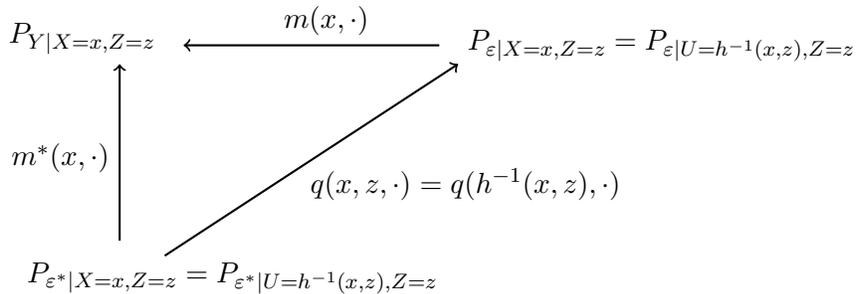

Therefore, varying $Z$ does not affect the distributions of the unobservables, but does affect $X$, which means that one can get at the exogenous effect of $X$ on $Y$. This is the same reasoning as in the case where $Z$ is absolutely continuous. In the binary case one can use a general sequencing argument as mentioned above. Let us be more specific about this sequencing argument now.

 Recall that in the univariate case \cite{torgovitsky2015identification} uses the measure preserving isomorphism $(x,z)\mapsto (F_{X|Z=z}(x),z)$ to condition $\varepsilon$ on $F_{X|Z}$ instead of $X,Z$ and then applies the monotone rearrangement $T(x)=F^{-1}_{X|Z=z'}(F_{X|Z=z}(x))$ as a map between $F_{X|Z=z}$ and $F_{X|Z=z'}$; this ensures that for every point $x_0\in\mathcal{X}_z\cup\mathcal{X}_{z'}$ $F_{\varepsilon|X=x_0,Z=z'}=F_{\varepsilon|X=Tx_0,Z=z}$ and analogously for $\varepsilon^*$, so that $q$ is the same for all iterations $T^n$. He then shows that in one dimension this iteration converges to a fixed point and can hence show that $q$ is constant for all starting points $x_0$, which by the assumed normalization implies that $m=m^*$.

Now, there are mainly two reasons for why this simple reasoning does not work in a higher dimensional setting. First, the map $(x,z)\mapsto (F_{X|Z=z}(x),z)$ is only invertible in the one-dimensional case, so this simple argument does not work. Our solution for this is Assumption \ref{normalizationass}. With this we can write $U=h^{-1}(X,Z)$ since both $P_{X|Z}$ and $P_U$ are absolutely continuous, so that $h$ is invertible, which gives 
\begin{equation}\label{equalitiesimportant}
P_{\varepsilon|X=x,Z=z}=P_{\varepsilon|U=h^{-1}(x,z),Z=z}=P_{\varepsilon|U=h^{-1}(x,z)}.
\end{equation} The second equality in \eqref{equalitiesimportant} follows from Assumption \ref{instrumentass}. The first equality follows from the following reasoning: the map $\phi: (X,Z) \mapsto (h^{-1}(X,Z),Z)$ is a measure-preserving isomorphism since $z\mapsto z$ is a measure preserving isomorphism and $x\mapsto h^{-1}(x,z)$ is a measure preserving isomorphism for all $z$, so that for every rectangle $E_x\times E_z\equiv(-\infty,x]\times (-\infty,z]\in\mathscr{B}_{\mathbb{R}^{k+m}}$ 
\[P_{X,Z}(E_x\times E_z) = P_{U,Z}(\phi^{-1}(E_x\times E_z))\equiv P_{U,Z}(h^{-1}(E_x,z)\times E_z).\] Analogously, the map $(\varepsilon,x,z)\mapsto (\varepsilon,h^{-1}(x,z),z)$ is a measure preserving isomorphism for the same reasoning so that for every rectangle $E_\varepsilon\times E_x\times E_z\equiv (-\infty,\varepsilon]\times (-\infty,x]\times (-\infty,z] \in\mathscr{B}_{\mathbb{R}^{d+k+m}}$ 
\[P_{\varepsilon,X,Z}(E_\varepsilon\times E_x\times E_z) = P_{\varepsilon,U,Z}(\phi^{-1}(E_\varepsilon\times E_x\times E_z))=P_{\varepsilon,U,Z}(E_\varepsilon\times h^{-1}(E_x,z)\times E_z).\]
Thus
\[P_{\varepsilon|X,Z}(E_\varepsilon)=\frac{P_{\varepsilon,X,Z}(E_\varepsilon\times E_x\times E_z)}{P_{X,Z}(E_x\times E_z)} =\frac{P_{\varepsilon,U,Z}(E_\varepsilon\times h^{-1}(E_x,z)\times E_z)}{P_{U,Z}(h^{-1}(E_x,z)\times E_z)}=P_{\varepsilon|U,Z}(E_\varepsilon).\] The last thing to notice is that conditioning on measure zero events does not cause issues, because $(X,Z) \mapsto (h^{-1}(X,Z),Z)$ is measurable with measurable inverse by definition of a measure preserving isomorphism, so that their $\sigma$-algebras coincide, i.e.~$\sigma(U,Z) = \sigma(X,Z)$. We give another formal proof of this fact in Lemma \ref{stringofequalitieslemma} in the appendix, using disintegrations.\footnote{Note that this conditioning is different from the approach in \citet{kasy2014instrumental}. Kasy used the mapping $\psi:(X,Z)\mapsto (X,h^{-1}(X,U))$, where he defined the inverse of $h$ is \emph{with respect to $X$}, which is not invertible since in his case it is a map from $\mathbb{R}^2$ to $\mathbb{R}\times \mathbb{U}$, where $\mathbb{U}$ is the (in Kasy's case possiby infinite dimensional) metric space containing $U$. Therefore, the respective $\sigma$-algebras $\sigma(X,Z)$ and $\sigma(X,h^{-1}(X,U))$ need not coincide. In our case, however, we use the measure preserving isomorphism $\phi(X,Z)= (h^{-1}(X,Z),Z)$, which is measurable with measurable inverse, so that $\sigma(X,Z)$ and $\sigma(h^{-1}(X,Z),Z)$ coincide.}

Second, we need to use a general sequencing argument which is more intricate in higher dimensions. By Assumption \ref{normalizationass} the distribution of $U$ is fixed to be $F_U=F_{X|Z=z}$ for one $z\in\mathcal{Z}$ and we require the map $h^{-1}(x,z)$ to be the identity and the other map $h^{-1}(x,z')$ to be the gradient of a convex function transporting $F_U$ onto $F_{X|Z=z'}$. They key step here then is a new result for the dynamics of measure preserving isomorphisms which take the form of the gradient of a convex function, which we prove in the appendix. Intuitively, we use Assumption \ref{supportass} and the strict monotonicity of the $F_{X|Z}$ to show that the gradient of a convex function $T$ mapping $F_{X|Z=z}$ onto $F_{X|Z=z'}$ never crosses $\mathcal{M}(z,z')$ in the sense that for each $x\in\mathcal{X}_z\cup\mathcal{X}_{z'}$ the curve $(1-t)x+tTx$, $t\in(0,1)$ never intersects $\mathcal{M}(z,z')$.
With this we can show that $\mathcal{M}(z,z')$ is a fixed set of the iteration $T^nx_0$ for the map $T$ between $F_{X|Z=z}$ and $F_{X|Z=z'}$. The key here is Assumption \ref{supportass} which requires that the manifold $\mathcal{M}(z,z')$ lies in the supports $\mathcal{X}_z=\mathcal{X}_{z'}$ in such a way that an iteration of this map converges to $\mathcal{M}(z,z')$ for every point. The proof of Theorem \ref{maintheorem2} contains the details. 

Note in this respect that if we were to make a different functional form assumption on $h$ we would have to work out the dynamics of a different measure preserving isomorphism which in turn would lead to a different Assumption \ref{supportass}. Gradients of convex functions are very general and well-behaved as transport maps, however, and it is not likely that one will find a measure preserving map with better properties. Even more importantly, the assumption that $h$ is the gradient of a convex function is the most natural generalization of a strictly increasing and continuous $h$ to the multivariate setting.
Lastly, Theorem \ref{maintheorem2abscont} in the appendix shows identification in the case where $Z$ is absolutely continuous under weaker assumptions on the supports $\mathcal{X}_z$, but requiring $m$ to be a measure preserving $C^{1}$-diffeomorphism instead of simply being a measure preserving isomorphism, which is much stronger. Its statement is a straightforward generalization of the result in \citet{torgovitsky2015supplement}. 

To conclude this section, we also want to stress again that Theorem \ref{maintheorem2} and its absolutely continuous counterpart from the appendix are \emph{not} constructive identification results in the sense that they do not provide us with the function $m$. They just provide the identified set $\mathbb{I}$ which we prove to contain a single element $m$, just as the univariate result \citet{torgovitsky2015identification} and \citet{d2015identification}. There are ways to estimate the function $m$ semi-parametrically like \citet{komunjer2010semi} or \cite{torgovitsky2016minimum}, but a fully nonparametric approach is still lacking. This is especially important to keep in mind in the following section where we show identification of the Hedonic model in multiple markets and identification of the BLP-model without index restrictions.

\section{Applications: BLP-, and Hedonic models}\label{applications}
In order to showcase the applicability of Theorem \ref{maintheorem2} we apply it in two different settings. First to the BLP \citep{berry1995automobile} model, where we complement the point-identification result of \citet{berry2014identification}. Second to Hedonic models with multivariate heterogeneity and endogenous characteristics, providing the first identification result in this setting and answering an open question posed in \citet{chernozhukov2014single} in the process.
Let us start with the former.

\subsection{The BLP model}
The BLP model \citep{berry1995automobile} was introduced for identifying and estimating utility- and cost functions of participants in demand and supply systems of differentiated product markets when only aggregate market share data are available to the researcher. The only article providing results on identification of the BLP model to date is the seminal \citet{berry2014identification}. In this article the authors need to make a somewhat artificial \emph{index restriction}, because their identification result relies on univariate identification results from the literature, in particular the identification result in \citet{chernozhukov2005iv}. Using Theorem \ref{maintheorem2} we can generalize their result directly to prove nonparametric identification without the need for the index restriction, complementing their result.
Let us start with the demand side.

\paragraph{Demand side} The demand side in this model is obtained by aggregating a continuum of individual discrete choice models in the following way, where we adapt the notation from \citet{berry2014identification}.
Each consumer $i$ in market $t$ chooses a good $j$ from a market $\mathcal{J}_t\coloneqq \{0,1,\ldots,J_t\}$, which consists of a continuum of consumers with total measure $M_t$. A market is formally defined by $(\mathcal{J}_t,\chi_t)$ with $\chi_t\coloneqq (x_t,p_t,\xi_t)$. Here, $x_t=(x_{1t},\ldots,x_{J_tt})$ is a $K\times J_t$ matrix containing the observed and exogenous characteristics of the products in the market. $\xi_t\coloneqq (\xi_{1t},\ldots,\xi_{J_tt})$ contains all of the unobservable characteristics at the product or market level and $p_t\coloneqq(p_{1t},\ldots,p_{J_tt})$ contains observable endogenous characteristics, i.e.~those characteristics which are correlated with $\xi_t$ like the price. 

Consumer preferences in the BLP model are determined by indirect utilities in the sense that consumer $i$ in market $t$ has conditional indirect utilities $v_{i0t},\ldots,v_{iJ_tt}$. Following \citet{berry2014identification}, we normalize the outside option $v_{i0t}$ to be zero, i.e.~$v_{i0t}=0$ for all $i$ and $t$ and assume that the utilities are independent and identically distributed across consumers and markets with joint distribution function
$F_{v}(v_{i1t},\ldots,v_{iJ_tt}|\chi_t).$ Then the standard assumption is that $\argmax_{j\in\mathcal{J}}v_{ijt}$ is unique with probability 1, which leads to the following definition of the market shares $s_{jt}$ for each product $j$ in market $t$:
\begin{equation}\label{berryequation}
s_{jt} = \sigma_j(\chi_t)=P\left(\argmax_{k\in\mathcal{J}} v_{ikt} = j | \chi_t\right),\quad j=0,\ldots, J,
\end{equation}
under the normalization $s_{0t} = 1-\sum_{k=1}^J s_{kt}$. Now here is where \citet{berry2014identification} are forced to introduce the \emph{index restriction} assumption, because it enables them to write the demand function element-wise for every $j$. In particular, they define a univariate index $\delta_{jt} = \delta_j(x_{jt},\xi_{jt})$ for each product $j$, where $\delta_j$ is a function which is strictly increasing and continuous in the unobservable $\xi_{jt}$ for $x_{jt}$.\footnote{In the main text, they even assume that $\delta_j$ is linear, a much stronger assumption, but they relax this assumption to allow for strictly increasing and continuous $\delta_j$ in the appendix, and this is the result we focus on.} The idea then is to write the demand function $\sigma_j(\chi_t)$ for each $j$ only in terms of $x_{jt}$, $\xi_{jt}$, and $p_t$, element-wise for every $j$, and then identify $\xi_{jt}$ by inverting $\sigma_j$ as well as the index $\delta_j$ to get
\[\xi_{jt} = \delta_j^{-1}\left(\sigma_j^{-1}(s_t,p_t),x_{jt}\right),\]
requiring both $\sigma_j$ and $\delta_j$ to be strictly increasing and continuous functions in $\xi_{jt}$.

It is exactly here where we can apply the framework from section \ref{torgovitskysubsection}. In fact, the functional form assumption of strict monotonicity and continuity on the index $\delta_j$ and the demand function $\sigma_j$ is the standard assumption from \cite{matzkin2003nonparametric}, and simply serves as a tool in order to work with a univariate unobservable for every distribution. Using Theorem \ref{maintheorem2}, we are able to prove identification of the model without being forced to make any index restrictions, therefore complementing the result in \citet{berry2014identification}. 
We do so as follows.

Firstly, we allow for a general demand function $\sigma$ solving the demand problem \eqref{berryequation} over all products $j = 1,\ldots, J$ in the market \emph{simultaneously} and for all individuals $i$ in the continuum $M_t$, so that $s_{t}=\sigma(\chi_t).$ This is the main difference to the index restriction: \citet{berry2014identification} allow for multiple products like we do, but they only do so element-wise, i.e.~they treat every good separately. We on the other hand can allow for general interactions between the products.
The outside option is still fixed as above. Analogous to \citet{berry2014identification}, we assume that this maximization problem has a unique solution. Note that this is the uniqueness assumption we need to make $\sigma$ determinable. In addition we have to assume that $\sigma$ is invertible between $\xi_t$ and $s_t$, a condition which might be hard to satisfy in practice, but has been the standard assumption in this literature (see \citet{matzkin2007heterogeneous} for further discussion of this point). \citet{berry2014identification} require monotonicity and continuity in every element, which is a sufficient condition for invertibility of $\sigma$ and might be even harder to satisfy in practice. 

Measure preservation is a natural assumption for BLP models. Recall that $s_t$ is the vector of market shares of every product, which possesses a certain (conditional) probability distribution $P_{s_t|x_t,p_t}$. This distribution is just the distribution of choices of individuals $i$ in the market $t$, so that one can view each point in the support of $s_t$ conditional on $x_t$ and $p_t$, denoted by $\mathcal{S}_{x_t,p_t}$, as the \emph{purchase plan} of an individual $i$, determining the probability with which this individual is to buy which product $j$ in the market. Then this individual $i$ needs to have a certain \emph{evaluation} of the products $j=1,\ldots,J$, which is unobservable to the econometrician, i.e.~a distribution over the unobservables $x_{jt}$ analogous to the purchase plan of the individual; this can be thought of as giving for each product $j$ a probability of how ``important'' the respective unobservable $x_{jt}$ of the product is for the individual's choice. 

Since all individuals lie on a continuum, it makes more sense to talk about \emph{sets of individuals} instead of unique individuals. Therefore, every (Borel-) set $E\in \mathcal{S}_{x_t,p_t}$ of individuals with purchase plans $P_{s_t|x_t,p_t}(E)$ must have a corresponding set in $P_{\xi_t}$ which is of the same size, because all evaluations and purchase plans are based on the same set of individuals. But this is exactly the definition of a measure preserving demand function $\sigma$, i.e.~we require 
\begin{equation}\label{BLPmeasurepreserving}
P_{s_t|x_t,p_t}(E) = P_{\xi_t}\left(\sigma^{-1}(x_t,p_t,E)\right)\quad \text{for all}\quad E\in\mathscr{B}_{s_t|x_t,p_t}.
\end{equation}

Now, again, as in the example of Hedonic models, the need for Theorem \ref{maintheorem2} arises from the fact that the $p_t$ are endogenous. To the best of our knowledge, this provides the first instance where nonseparable triangular models can be used for identification of the BLP model. Those results were not possible previously, because they required that the second stage in those models be univariate for point-identification, as argued in \citet[p.~1754]{berry2014identification}. Providing complete identification of the BLP model therefore provides another instance proving how important a multivariate generalization of these identification results really is. 
In order to apply Theorem \ref{maintheorem2}, we need to model the first stage relationship
\begin{equation}\label{BLPfirststage}
p_t = h(z_t, u_t),
\end{equation}
where $z_t$ is a set of instruments which are excluded from the demand model, $u_t$ is a vector of unobservables with non distribution and of the same dimension as $p_t$ and $h$ is a measure preserving isomorphism, which we assume to be the gradient of a convex function transporting the distribution of $u_t$ onto the distribution of $p_t$ for all $t$. 

In our setting it is very natural to let $p_{jt}$ be multivariate for every $j$, hence letting $p_t$ be a matrix. All we have to do to make this work is to \emph{vectorize} the matrix $p_t$ by stacking each column onto one another, i.e.~identifying the matrix space $\mathbb{R}^{J_t\times K}$ with $\mathbb{R}^{J_t\cdot K}$, where $K$ is the number of columns for every $p_{jt}$. Note again, that we can allow for instruments $z_t$ which are discrete and lower dimensional than the endogenous variables $p_t$, allowing for binary policy changes. All the instruments need to satisfy is $z_t\perp (u_t,\xi_t)$ and that they have an influence on each element of $p_t$. 
Let us now state the identification result for the demand side.
The observables of the market are $(M_t,x_t,p_t,s_t,z_t)$. The demand side is modeled through \eqref{BLPmeasurepreserving} and \eqref{BLPfirststage}. Then the following holds.
\begin{proposition}[General identification of the demand side in the BLP model]\label{BLPdemandproposition}
In the case where the instruments $z_t$ are absolutely continuous, let the regularity assumptions hold as stated in Theorem \ref{maintheorem2abscont} in the appendix. In the case where the instruments $z_t$ are discrete, let Assumptions \ref{dimensionsass} -- \ref{supportass} hold for $p_t$, $z_t$, $u_t$, and $\xi_t$. Then the model \eqref{BLPmeasurepreserving} and \eqref{BLPfirststage} is identified in the sense that the identified set 
\[\mathbb{I}\coloneqq \{\sigma\in\mathcal{H}(S_{x_t,p_t},\xi_t):(\sigma^{-1}(x_t,p_t,s_t),u_t)\perp z_t\}\] contains an almost everywhere unique element $\sigma$. $\mathcal{H}(S_{x_t,p_t},\xi_t)$ is the set of all isomorphisms between $\xi_t$ and $s_t$ for exogenous $x_t$ and $p_t$. 
\end{proposition}
The proof of this proposition follows immediately from Theorem \ref{maintheorem2} or Theorem \ref{maintheorem2abscont}. Proposition \ref{BLPdemandproposition} therefore provides nonparametric identification for the demand side of the BLP model in the most general case, only requiring $\sigma$ to be a measure preserving isomorphism. Note that Assumption \ref{mpiso2} requires a normalization of the demand function. This can be done by assuming a multivariate uniform distribution for $\xi_t$, in the sense that all $\xi_{jt}$ are uniformly distributed, which is the analogue to the normalization in \citet{berry2014identification} who assume a univariate uniform distribution for every $\xi_{jt}$. In our case, one is actually free to model the dependency structure between the $\xi_{jt}$, however, i.e.~one is not required to assume that they are all independently distributed. As for $\sigma$ being a measure preserving isomorphism, this is satisfied as soon as the utility maximization problem has a unique and invertible solution.

The nice thing about Proposition \ref{BLPdemandproposition} is that the assumptions of a unique and measure preserving demand function $\sigma$ are natural and can be implied by the set-up of the model. Also notice how our approach allows for multivariate $p_t$ even from the set-up. The last interesting and also important thing to recall is that we can allow for instruments to be of lower dimension than the endogenous variables $p_t$. This is especially important in practice. In fact, it might often be the case that there is a dichotomous shock introduced into the model, possibly through a policy change, which can serve as an instrument.
If this policy change is truly independent and is such that it influences all $p_t$, then it alone can serve as a single instrument to identify the whole demand side of the model, under the restriction that Assumption \ref{supportass} on the supports of $P_{p_t|z_t=z}$ and $P_{p_t|z_t=z'}$ is satisfied; but this assumption can be checked in higher dimensions and simply be eyeballed in the case where $p_t$ is two-dimensional, an important special case. 

\paragraph{Supply side} Having identified the demand side, one can model the supply side in basically two ways. First, one can simply assume that one \emph{knows} the oligopolistic structure of the supply side in which case one can immediately deduce the vector of marginal costs $mc_t\coloneqq (mc_{1t},\ldots,mc_{J_tt})$ by
\begin{equation}\label{supplyside1}
mc_t = \psi(s_t,M_t,\sigma,p_t),
\end{equation} since all quantities on the right hand side are observed ($s_t$, $M_t$, $p_t$) or identified ($\sigma$). Note that \eqref{supplyside1} is more general than the function proposed in \citet{berry2014identification}, which is, again, only defined element-wise, i.e.~one $\psi_j$ for every product $j$, analogous to their element-wise definition of the demand function $\sigma_j$. In the case for known $\psi$, there is nothing to do from an econometric perspective, as one simply assumes away the problem of identifying the respective monopoly structure, i.e.~the function $\psi$. This can be warranted in some cases, where one has additional knowledge on the oligopoly structure. Based on this, one can identify the cost functions 
\begin{equation}\label{supplyside2}
mc_t = c(Q_t,w_t,\omega_t)
\end{equation}
with some instrument (i.e.~supply shifter) and Theorem \ref{maintheorem2} or Theorem \ref{maintheorem2abscont}.
Here, $w_t\coloneqq(w_{1t},\ldots,w_{J_tt})$ are observable and exogenous cost shifters, and $\omega_t\coloneqq(\omega_{1t},\ldots,\omega_{J_tt})$ are unobservable cost shifters. Let $\mathcal{J}_j$ denote the set of products produced by the firm producing product $j$. Let $q_{jt}=M_ts_{jt}$ be the quantity produced of good $j$ in equilibrium and let $Q_{jt}$ be the vector of quantities of all goods $k\in\mathcal{J}_j$. This setting is completely analogous to the demand side if we replace $w_t\equiv x_t$, $\omega_t\equiv\xi_t$, and $Q_t=p_t$. 
 
The important and more realistic way to model the supply side, however, is to allow for an unknown function $\psi$. Note that in this case, there are several approaches towards identification. In principle, there are four things to identify in the model: $mc_t$, $\psi$, $c$, and the unobservable shocks $\omega_t$. The approach in \citet{berry2014identification} is to combine \eqref{supplyside1} and \eqref{supplyside2} into one equation
\begin{equation}\label{supplysidecombined}
\omega_t = c^{-1}\left(Q_t,w_t,\psi(s_t,M_t,\sigma,p_t)\right)\coloneqq \pi^{-1}(p_t,M_t,s_t,w_t),
\end{equation}
eliminating $mc_t$ in the process. This approach enables them to identify the unobservable shocks $\omega_t$ as well as the function $\pi$ which incorporates both $c$ and $\psi$. Their approach consists of assuming that both $c$ and $\psi$ can be written element-wise, with each $c_j$ being linear in $w_{jt}$ and $\omega_{jt}$. We can do the same but much more generally with our approach, simply requiring $\psi$ and $c$ to be measure preserving isomorphisms. Then $\pi^{-1}$ must be a measure preserving isomorphism, too. Now realize that \eqref{supplysidecombined} is perfectly analogous to \eqref{BLPmeasurepreserving}. Therefore, we can apply Theorem \ref{maintheorem2} again in order to establish identification. We only need some instrument $z_t$ for the supply side. Then a first stage for the endogenous $p_t$ is 
\begin{equation}\label{firststagesupply}
p_t = h(z_t,u_t),
\end{equation}
where, again, $u_t$ is unobservable and of the same dimension as $p_t$ and $h$ is the gradient of a convex function, exactly as in \eqref{BLPfirststage}. 
Then we can state the following identification result for the supply side, which we only state in terms of Theorem \ref{maintheorem2}---the statement for Theorem \ref{maintheorem2abscont} is of course analogous.
\begin{proposition}\label{BLPsupplyproposition}
In the case where the instruments $z_t$ are discrete, let Assumptions \ref{dimensionsass} -- \ref{supportass} hold for $p_t$, $z_t$, $u_t$, and $\omega_t$. Then the model consisting of \eqref{supplysidecombined} and \eqref{firststagesupply} is identified in the sense that the identified set 
\[\mathbb{I}\coloneqq \{\pi\in\mathcal{H}(S_{w_t,M_t,p_t},\omega_t):(\pi^{-1}(w_t,M_t,p_t,s_t),u_t)\perp z_t\}\] contains an almost everywhere unique element $\pi$. $\mathcal{H}(S_{w_t,M_t,p_t},\omega_t)$ is the set of all isomorphisms between $\omega_t$ and $s_t$ for exogenous $w_t$ and $p_t$ and fixed $M_t$. 
\end{proposition}
The proof again follows immediately from Theorem \ref{maintheorem2}. Again, one has to make a normalization assumption, but assuming that $mc_t$ is multivariate uniform is not a strong restriction. 

Now, in many cases one is actually interested in identifying both functions $c$ and $\psi$ \emph{separately}, because $\psi$ gives information about the oligopoly structure of the market. This can again be done in our setting if one has proper instruments for both equations, which in general is not a strong requirement, because one can use standard exogenous shifters as argued in \citet{berry2014identification}.
The idea is to use our identification approach first on \eqref{supplyside1} and then on \eqref{supplyside2} once we have identified $mc_t$. Let us assume that there are appropriate sets of shifters (i.e.~instruments) $z_t^1$ and $z_t^2$ for \eqref{supplyside1} and \eqref{supplyside2}, respectively. Note that we do admit the likely case $z^1_t=z^2_t$. We also need two first stage equations:
\begin{equation}\label{supplyside1first}
p_t = h_1(z^1_t,u^1_t),
\end{equation}
and
\begin{equation}\label{supplyside2first}
Q_t = h_2(z^2_t,u^2_t),
\end{equation}
for unobservables $u_t^1$ and $u_t^2$ with the same dimension as $p_t$ and $Q_t$, respectively, and $h_1$ and $h_2$ are gradients of convex functions. Since $Q_t$ is a matrix, we rely on the simple trick of writing it as a vector, stacking the columns upon one another, as mentioned. 
 The main identification result for the supply side is then as follows. Again, we only state the result for Theorem \ref{maintheorem2}, but the result for Theorem \ref{maintheorem2abscont} is perfectly analogous.
\begin{proposition}[General identification of the supply side of the BLP model]\label{BLPsupplyproposition}
In the case where the instruments $z^1_t$ and $z^2_t$ are discrete, let the first stages \eqref{supplyside1first} and \eqref{supplyside2first} satisfy Assumptions \ref{dimensionsass} -- \ref{normalizationass}, let $c$ and $\psi$ satisfy Assumption \ref{mpiso2}, and let Assumptions \ref{instrumentass} as well as \ref{abscont} -- \ref{supportass} hold for $p_t$, $Q_t$, $z^1_t$, $z^2_t$, $u^1_t$, $u_t^2$, $mc_t$, and $\omega_t$. Then the two models consisting of \eqref{supplyside1} and \eqref{supplyside1first} as well as \eqref{supplyside2} and \eqref{supplyside2first} are identified in the sense that the identified sets
\[\mathbb{I}_1\coloneqq \{\psi^{-1}\in\mathcal{H}(S_{w_t,M_t,p_t},mc_t):(\psi(\sigma,M_t,p_t,s_t),u_t^1)\perp z^1_t\}\] 
and 
\[\mathbb{I}_2\coloneqq \{c\in\mathcal{H}(S_{w_t,M_t,Q_t},\omega_t):(c^{-1}(w_t,M_t,Q_t,s_t),u_t^2)\perp z^2_t\}\] 
contain almost everywhere unique elements $\psi$ and $c$, respectively. As before, $\mathcal{H}(S_{w_t,M_t,p_t},mc_t)$ and $\mathcal{H}(S_{w_t,M_t,Q_t},\omega_t)$ are the respective sets of isomorphisms. 
\end{proposition}
\begin{proof}
First, one needs to identify $mc_t$ and $\psi$ in \eqref{supplyside1} and \eqref{supplyside1first}. This follows immediately from Theorem \ref{maintheorem2} and the fact that $\sigma$ has already been established to be identified on the demand side by Proposition \ref{BLPdemandproposition}. Then once $mc_t$ is identified, one can turn to the identification of $c$ and $\omega_t$ in  \eqref{supplyside2} and \eqref{supplyside2first}, which also follows immediately from Theorem \ref{maintheorem2} and the fact that $mc_t$ has already been identified.
\end{proof}
Propositions \ref{BLPdemandproposition} and \ref{BLPsupplyproposition} together establish general identification of the BLP model, without the need to make index restriction assumptions  and hence allowing for the most general set-up. The main assumptions, in addition to regularity assumptions like continuity and invertibility, guaranteeing this result are measure preservation and uniqueness of the demand function $\sigma$ as well as the cost functions $\psi$ and $c$. Both are very natural and fundamental, being a direct result of the model set-up. As a result of the endogeneity of quantity and price, we rewrote both sides of the model as a nonseparable triangular system and applied our two results from the previous section to guarantee point identification of the respective functions. There are certainly other ways for identification of this model than using Theorem \ref{maintheorem2} or Theorem \ref{maintheorem2abscont}, but both are very general and it is rather unlikely that even more general theorems will hold for this setting. Note that this application was just an outline that a different identification result to \citet{berry2014identification} holds, where one does not have to make index restriction assumptions or assume that the demand function is element-wise strictly increasing and continuous, but can allow for more general results. The discussion in this section was very high level. A next step from here is to find appropriate low level assumptions on the demand function (other than monotonicity) which imples that it is measure-preserving and invertible. Those should come from economic theory. 
Let us now turn to the main application of our main result where we prove nonparametric identification of Hedonic models with multivariate unobservable heterogeneity and endogenous characteristics.

\subsection{Hedonic models with endogenous characteristics}
Theorem \ref{maintheorem2} enables us to prove nonparametric identification of general Hedonic models with multivariate unobservable heterogeneity and endogenous observable characteristics in a multi-market setting, the main result of this whole section. 

Seminal identification results in \citet{ekeland2004identification} and \citet{heckman2010nonparametric} have focused on single-market Hedonic models with univariate unobservable heterogeneity and stronger functional form assumptions on the utility function. The recent working paper \citet{chernozhukov2014single} extended these results to single-market Hedonic models with multivariate unobservable heterogeneity using optimal transport theory. We complement these results in this subsection by for the first time allowing for endogenous observable characteristics in a multi-market setting under multivariate unobservables. 

We adapt the notation in \citet{chernozhukov2014single} slightly to make it compatible with the notation from the previous section. We consider an environment where consumers and producers trade a good or contract which is fully characterized by its type or quality $y\in\mathcal{Y}\subseteq\mathbb{R}^d$. Its price, $p(y)$ is determined endogenously in equilibrium. Producers $\Wtilde$ and consumers $\Xtilde$ are characterized by their respective types, $\wtilde\in\widetilde{\mathcal{W}}\subseteq\mathbb{R}^{m+d}$ and $\xtilde\in\widetilde{\mathcal{X}}\subseteq\mathbb{R}^{k+d}$. They are price takers and maximize quasi-linear utility $U(\xtilde,y)-p(y)$ and profit $p(y)-C(\wtilde,y)$, respectively, where $U$ is upper semicontinuous and bounded and $C$ is lower semicontinuous and bounded. Both are normalized to zero in case of nonparticipation.
We use the following equilibrium concept.
\begin{assumption}[Equilibrium concept]\label{equilibriumconcept}
The pair $(\gamma,p)$, for $\gamma$ a probability measure on $\widetilde{\mathcal{X}}\times\mathcal{Y}\times\widetilde{\mathcal{W}}$ and $p$ a function on $\mathcal{Y}$, is a hedonic equilibrium in the sense that $\gamma$ has marginals $P_{\Xtilde}$ and $P_{\Wtilde}$ and for $\gamma$-almost all $(\xtilde,y,\wtilde)$
\begin{align*}
U(\xtilde,y)-p(y) = \max_{y'\in\mathcal{Y}}(U(\xtilde,y')-p(y)),\\
p(y)-C(\wtilde,y) = \max_{y'\in\mathcal{Y}}(p(y)-C(\wtilde,y')).
\end{align*}
In addition, observed qualities $y=y(\xtilde,\wtilde)$, which maximize the joint surplus $U(\xtilde,y)-C(\wtilde,y)$ for all $(\xtilde,\wtilde)\in\widetilde{\mathcal{X}}\times\widetilde{\mathcal{W}}$, lie in the interior of $\mathcal{Y}$.
\end{assumption}
For absolutely continuous distributions $P_{\Xtilde}$ and $P_{\Wtilde}$ \citet{ekeland2010existence} and \citet{chiappori2010hedonic} show that a pure equilibrium exists and is unique under the twist condition from optimal transport \citep[p.~216]{villani2008optimal} which requires that the gradients $\nabla_{\xtilde}U(\xtilde,y)\coloneqq\frac{\partial}{\partial \xtilde} U(\xtilde,y)$ and $\nabla_{\wtilde} C(\wtilde,y)\coloneqq\frac{\partial}{\partial \wtilde} C(\wtilde,y)$ exist and are injective as functions of quality $y$. 
To obtain an estimable model, unobservable heterogeneity is introduced in the following way.\footnote{We only focus on the consumer's side throughout, because identification of the supplier's side is analogous.}
\begin{assumption}[Unobservable heterogeneity]\label{unobservedheterogeneity}
The observable type $\xtilde$ consists of an observable portion $x\in\mathbb{R}^k$ as well as a random unobservable portion $\varepsilon\in\mathbb{R}^d$, i.e.~$\xtilde=(x,\varepsilon)$.\footnote{Recall that $\varepsilon$ is assumed to be of the same dimension as $Y$.} Furthermore, utility $U(\xtilde,y)$ can be decomposed as $U(\xtilde,y) = \bar{U}(x,y) + \xi(x,\varepsilon,y)$.
\end{assumption}
In this setting, the object of interest for identification is the deterministic component of the utility function, $\bar{U}(x,y)$. Denote $V(x,y) = p(y)-\bar{U}(x,y).$
Here is where \citet{chernozhukov2014single} need to make the rather strong assumption that the observable characteristics $X$ are exogenous, i.e.~$X\independent\varepsilon$. The idea for identification then proceeds roughly as follows. The Monge-Kantorovich problem leads to natural functional form restrictions in this setting as the consumer's optimization program is to choose $y$ such that
\begin{equation}\label{generalconjugate}
V^\gamma=\sup_y\left\{\xi(x,\varepsilon,y)-V(x,y)\right\}.\footnote{Note that this optimization problem is the definition of the convex conjugate \citep[\S 12]{rockafellar1997convex} if $\xi(x,\varepsilon,y) = y'\varepsilon$, as in this case
$V^*(x,y) = \sup_{y\in\mathcal{Y}}\left\{y'\varepsilon-V(x,y)\right\},$ where $V^*(x,y)$ is the convex conjugate of $V(x,y)$.}
\end{equation}

Identification of $\bar{U}(x,y)$ requires identification of $V(x,y)$; to identify the latter one can use two approaches. One way is to show that the pair $(V(x,y),V^\gamma(x,y))$ uniquely solves the dual problem of the Monge-Kantorovich problem under the cost function $\xi(x,\varepsilon,y)$ in the general case under some regularity assumptions. The second route to identify $V(x,y)$ is via the optimal planner's problem, which takes the form of the Monge-Kantorovich problem with general cost function $\xi(x,\varepsilon,y)$. In the special case $\xi(x,\varepsilon,y)=y'\varepsilon$ considered in \citet{ekeland2004identification} for example, the cost function for the Monge-Kantorovich problem is the squared Euclidean norm, in which case the inverse demand function $m^{-1}(x,y)$ to take the form of the Brenier map as the determinable measure preserving map between $P_{Y|X}$ and $\varepsilon$ under regularity assumptions. To root this set-up in our notation from above, note that $m^{-1}(x,y)$ is the inverse of the measure preserving map $m:\mathcal{E}\to\mathcal{Y}_x$ from the setting in section \ref{mainidentsection}.

Allowing for general $\xi(x,\varepsilon,y)$ is a generalization of the concept of convex duality, leading to \eqref{generalconjugate}, see \citet[Chapter 2.4]{villani2003topics}. \eqref{generalconjugate} possesses a dual problem of the form
\begin{equation}\label{generalconjugate2}
(V^\gamma)^\gamma=\sup_\varepsilon\left\{\xi(x,\varepsilon,y)-V^\gamma(x,y)\right\},
\end{equation}
analogous to convex duality. For standard convex duality, if $V$ is convex and proper, it holds that $(V^*)^*=V$ \citep[\S12]{rockafellar1997convex}. Therefore, if the solution $(V^\gamma)^\gamma$ to the dual problem \eqref{generalconjugate2} coincides with $V$, we say that $V$ is $\gamma$-convex. This idea is also used in \citet{chernozhukov2014single} in the following assumption, which we also require.
\begin{assumption}\label{gammaconvex}
$V(x,y)$ is $\gamma$-convex.
\end{assumption}

Under some further regularity assumptions \citet{chernozhukov2014single} then use the fact that the Monge-Kantorovich problem has a unique measure preserving solution, so that $m^{-1}(x,y)$ is determinable in the sense of Definition \ref{uniquelyidentifiable}. Note that it is the assumption of exogenous $X$ which enables them to identify $m^{-1}(x,y)$, because they use the functional form restrictions imposed by the Monge-Kantorovich on $m^{-1}(x,y)$ to identify it, generalizing the identification result from \citet{matzkin2003nonparametric} in the process. They cannot allow for endogenous $X$, however, and explicitly leave open the problem of identifying it. This is where we can apply Theorem \ref{maintheorem2}.

For the endogenous case instruments in the form of exogenous shifters are required. One very convenient setting exists if the researcher has access to data in several, i.e.~at least two, disjoint markets $z$ and $z'$. This is where Theorem \ref{maintheorem2} shines, because one can consider the two markets as realizations of a binary instrument $Z$ under the assumption that $\varepsilon$ does not change between markets, i.e.~$Z \independent \varepsilon$. Then if the distribution of $X$ changes between the two markets, $Z$ is an instrument for $X$, as needed for Theorem \ref{maintheorem2}. Formally, we have to include a first stage $X=h(Z,\eta)$, where $h$ satisfies Assumptions \ref{mpiso} and \ref{normalizationass} for every market $z_i$ and the unobservable $\eta$ is of the same dimension as $X$.

\begin{assumption}[Multiple markets as a discrete instrument]\label{marketsassumption}
The distribution of the unobservable $\varepsilon$ does not change between different markets $z$, $z'$, but the distribution of $X$ does.
\end{assumption}

As before, we assume that $F_\eta$ is known and $h$ is the gradient of a convex function, as done in the proof of Theorem \ref{maintheorem2}.  
Now in order to apply Theorem \ref{maintheorem2}, we need to not only guarantee that $m^{-1}(x,y)$ is the determinable measure preserving map between $P_{Y|X=x,Z=z_i}$ and $P_{\varepsilon|X=x,Z=z_i}$, but also that it is continuous in all $x$. It turns out that under reasonably weak regularity assumptions we cannot guarantee this strong form of continuity; we can guarantee continuity in measure, however, so that we have to use Assumption \ref{mpiso2prime}. 

\begin{assumption}[Regularity assumptions]\label{regularityass} The following hold:
\begin{itemize}
\item[(i)] $U(\tilde{x},y)$ and $C(\tilde{y},y)$ are differentiable with respect to $x$ and $w$, respectively, for all $\tilde{x}$ and $\wtilde$.
\item[(ii)] $\xi(x,\varepsilon,y)$ is continuously differentiable with respect to $y$ for all $x$, $\varepsilon$, and $y$.
\item[(iii)] $\det\left(\nabla_{y\varepsilon}^2\xi(x,\varepsilon,y)\right)\neq0$ for all $x,\varepsilon,y$, where $\nabla_{y\varepsilon}^2\equiv (\partial^2\xi / \partial\varepsilon_i\partial y_j)_{ij}$ denotes the Hessian.
\item[(iv)] Twist condition: For all $x$ and $y$, the gradient $\nabla_y\xi(x,\varepsilon,y)$ of $\xi(x,\varepsilon,y)$ in $y$ is injective as a function of $\varepsilon$.
\item[(v)] $h$ is the gradient of a convex function transporting $\eta$ onto $X$.
\item[(vi)] $P_{\varepsilon}$ is known.
\item[(vii)] Assumptions \ref{dimensionsass} -- \ref{supportass} hold for this model.
\end{itemize}
\end{assumption}
Parts (i) -- (iv) of Assumption \ref{regularityass} guarantee existence and uniqueness of an optimal transport map between $P_{Y|X=x,Z=z_i}$ and $P_{\varepsilon|X=x,Z=z_i}$ for all $x$ and $z_i$ (see e.g.~Chapter 10 in \citeauthor{villani2008optimal} \citeyear{villani2008optimal}) and are the same as in \citet{chernozhukov2014single}. Parts (v) -- (vii) are the assumptions we need to require for Theorem \ref{maintheorem2}. Overall, our regularity assumptions are only slightly more demanding than the ones of \citet{chernozhukov2014single} in the exogenous case, but we can allow for endogenous characteristics. 
Of course, we need to require that $F_{X|Z=z_i}$ satisfy Assumption \ref{supportass}, which requires that their supports are convex, do not change between markets, and admit an appropriate manifold $\mathcal{M}(z_i,z_j)$. This leads to the following result.

\begin{theorem}[Identification of utility functions]\label{utilityidentprop}
In the Hedonic model defined in Assumptions \ref{equilibriumconcept} and \ref{unobservedheterogeneity} assume $\varepsilon\not\independent X$. Let the researcher have access to data in at least two disjoint markets $z$ and $z'$ and assume a first stage of the form $X=h(Z,\eta)$ for $\eta,X\in\mathbb{R}^k$. Furthermore, let Assumptions \ref{gammaconvex}, \ref{marketsassumption}, and \ref{regularityass} hold. Then the utility function $\bar{U}(x,y)$ is identified for almost every $(x,y)\in\mathcal{X}\times\mathcal{Y}$ up to an additive constant.
\end{theorem}

The idea for the proof of Theorem \ref{utilityidentprop} is to apply Theorem \ref{maintheorem2}. As stated, the first four parts of Assumption \ref{regularityass} guarantee existence and uniqueness of the optimal transport map $m^{-1}(x,y)$ as proved in \citet{chernozhukov2014single}, which is the determinable measure preserving map between $P_{Y|X=x,Z={z_i}}$ and $P_{\varepsilon|X=x,Z=z_i}$ required by Theorem \ref{maintheorem2}. The additional requirement that this measure preserving map be continuous in measure for every $x$ is guaranteed by the existence of a continuous disintegration and a slight generalization of stability results for the Monge-Kantorovich problem \citep[Chapter 5]{villani2008optimal}. The other regularity assumptions then guarantee that $m^{-1}(x,y)$ can be identified for almost every $y$ and $x$ through Theorem \ref{maintheorem2}. Assumption \ref{gammaconvex} is needed to guarantee differentiability of $V(x,y)$, which in turn leads to the identification of $\bar{U}(x,y)$ for almost every $x$ and $y$ through a first order condition since $m^{-1}(x,y)$ is identified for almost every $x$ and $y$. The detailed proof is in the appendix.

Theorem \ref{utilityidentprop} is the first result in the literature to use data from multiple markets to identify Hedonic models with multivariate unobservable and endogenous characteristics; it therefore answers the open question stated in \citet{chernozhukov2014single} who ask under which conditions one can identify Hedonic models with multivariate unobservables that are not independent of the observables. 

\section{Conclusion}\label{conclusion}
In this article we have proposed a framework for nonparametric point-identification of nonseparable triangular models with a multivariate first- and second stage.
The main result is a direct generalization of the seminal results from \citet{torgovitsky2015identification} and \citet{d2015identification} for point-identification of nonseparable triangular models with discrete instruments. In particular, we derive primitive conditions on the data generating process under which a nonseparable triangular model with a multivariate first and second stage is identified.

This result is widely applicable. In fact, we use it to derive assumptions under which both the supply and the demand side of the BLP model are nonparametrically identified, even under general heterogeneity. Previously, one had to uphold index restrictions \citep{berry2014identification}. As a second  and main application, we prove the first nonparametric identification result for Hedonic models with endogenous characteristics and multivariate heterogeneity, treating different markets as realizations of a discrete instruments. In particular, this answers an open question in \citet{chernozhukov2014single} showing when Hedonic models with general heterogeneity are identified under endogeneity. Other possible applications not addressed in this article are to competing risk models \citep{lee2013nonparametric} or generalized random coefficient models \citep{lewbel2017unobserved}. 

We were able to obtain the theoretical result, because we use and derive some apparently new results in the theory of optimal transport, in particular, we derive some apparently new results about properties of transport maps which take the form of gradients of convex functions, the arguably most natural generalization of a strictly increasing and continuous function to the multivariate setting. In particular, we prove a new result on their dynamics between two absolutely continuous measures whose supports coincide, and defining a criterion for the existence of a fixed set of measure zero. This result builds the heart of our third main theoretical result, but is also of interest in itself as it for instance can be applied to provide a new explanation for how equilibria are formed in General Equilibrium theory. 

Finally, our main theoretical identification result is non-constructive as it is a direct generalization of the seminal result of \citet{torgovitsky2015identification}. There are ways to use this identification result for semi-parametric identification, but a fully nonparametric approach is still lacking. The next important step will hence be to find slightly stronger assumptions than the current ones which would enable nonparametric estimation and inference in these models. Moreover, since the identification result rely upon optimal transport theory, it is imperative to derive further statistical properties of these maps, in particular their large sample distributions.

\bibliography{nonseparable}
\appendix
\section{The Monge-Kantorovich problem}\label{mongekantorovichsection}
A way to model determinable measure preserving maps is via the theory of Optimal Transport, a very active research area in (Applied) Mathematics with important recent applications in economics and econometrics. For references on this vast subject we refer to \citet{gangbo1996geometry}, \citet{rachev1998mass}, \citet{villani2003topics}, \citet{villani2008optimal}, and \citet{santambrogio2015optimal}; the recent addition to this literature, \citet{galichon2016optimal}, approaches the subject through the lens of economics. 

The original goal of this area of research is to find a measure preserving map $T$ which transports one probability measure $P_\varepsilon$ onto another probability measure $P_Y$ in a ``cost-efficient way''. The set up for this is the \emph{Monge-Kantorovich problem}. 
To be precise, the Monge- and the Kantorovich problem are actually two different problems, the latter being the convex relaxation of the former. Monge's Problem asks for an optimal transport map between two (probability) measures, $P_{\varepsilon}$ and $P_{Y}$, where optimality is measured with respect to some cost function $c:\mathcal{E}\times\mathcal{Y}\to\mathbb{R}$. 
This problem can be stated as
\begin{equation}\label{monge}
\text{minimize} \int_{\mathbb{R}} c\left(e,T(e)\right) dP_{\varepsilon}(e)\quad T:\mathcal{E}\to\mathcal{Y}\thickspace\text{measurable}.
\end{equation}
Here $y\in\mathcal{Y}$, $e\in\mathcal{E}$, and $\Pi(P_{Y},P_{\varepsilon})$ is the set of all probability measures on $\mathcal{Y}\times\mathcal{E}$ such that the marginal distributions of some $\pi\in\Pi(P_{Y},P_{\varepsilon})$ are precisely $P_{Y}$ and $P_{\varepsilon}$. 
The Kantorovich problem between (probability) measures $P_{Y}$ and $P_{\varepsilon}$ under some cost function $c:\mathcal{Y}\times\mathcal{E}\to\mathbb{R}$ asks for an optimal transport \emph{plan} in the sense that the transport does not have to be accomplished through a function as in the Monge problem, but is concentrated on the support $\Gamma$ of a joint probability distribution $\gamma$ which has $P_\varepsilon$ and $P_Y$ as marginals.
\begin{equation}\label{kantorovichproblem}
\min_{\pi\in \Pi(P_{Y},P_{\varepsilon})} \int_{\mathcal{Y}\times\mathcal{E}} c(y,e)d\pi(y,e).
\end{equation}
To make the two problems more tangible, one can picture a pile of sand and some hole of the same volume. Then the Kantorovich problem asks for the most cost effective way to put the sand into the hole, allowing each grain of sand to be split up further. 
The Monge problem requires the optimal transportation plan mapping $P_{\varepsilon}$ to $P_{Y}$ to be such that no grain of sand be split, i.e.~that the transport can actually be accomplished through some function $T$. 
For many cost functions $c$, the solution to the Monge and the Kantorovich problem actually coincide under the assumption that $P_\varepsilon$ is absolutely continuous, so that it is legitimate to speak of the Monge-Kantorovich problem in these cases. Moreover, this solution is \emph{unique} for many important cost functions. All of these statements can be found in Chapter 1 of \cite{villani2003topics}. 

Unique solutions of the Monge-Kantorovich problem are hence determinable measure preserving maps.
The convenience of the Monge-Kantorovich problem lies in the fact that different cost functions $c$ lead to different measure preserving maps and even isomorphisms, many of which occur naturally in economics and econometrics. In fact, Optimal Transport theory has found many applications in optimal matching theory (see the applications in \citeauthor{galichon2016optimal} \citeyear{galichon2016optimal}). The arguably most important measure preserving map which can be derived from the Monge-Kantorovich problem is the Brenier map, which has already seen some application in Statistics (\citeauthor{carlier2016vector} \citeyear{carlier2016vector}; \citeauthor{carlier2016misspecified} \citeyear{carlier2016misspecified}; \citeauthor{chernozhukov2014single} \citeyear{chernozhukov2014single}; \citeauthor{chernozhukov2016monge} \citeyear{chernozhukov2016monge}) as a multivariate generalization of strictly increasing and continuous functions. Let us give a brief overview. 

\subsection{Brenier map} 
The Brenier map results from solving the Monge-Kantorovich problem under the standard squared Euclidean distance as a cost function, i.e.
\[c(y-x)=\|y-x\|^2_2=\sum_{i=1}^d |y_i-x_i|^2.\]

\cite{brenier1991polar} first proved that if $P_\varepsilon$ and $P_{Y|X=x}$ possess finite second order moments and if $P_\varepsilon$ is absolutely continuous, then the Monge and the Kantorovich problem coincide, and the \emph{unique} solution to 
\begin{equation}
T_0\coloneqq\argmin_T \int_{\mathcal{E}} \|e-T(e)\|_2^2 dP_\varepsilon(e),\quad T\thickspace\medspace\text{measurable}
\end{equation}
is the \emph{gradient of a convex function}, i.e.~$T_0(e)=\nabla\varphi(e)$ for some convex $\varphi$. Based on this result, gradients of convex functions are usually called Brenier maps in the Optimal Transport literature. Later, \citet{mccann1995existence} proved that the Brenier map always exists \and{is unique} between two probability measures $P_\varepsilon$ and $P_{Y|X=x}$ \emph{as soon as} $P_\varepsilon$ \emph{is absolutely continuous}, without the assumption of finite second order moments:
\begin{theorem}[\citeauthor{mccann1995existence} \citeyear{mccann1995existence}]\label{mccann}
Let $P_\varepsilon,P_Y$ be two Borel probability measures on $\mathbb{R}^d$ such that $P_\varepsilon$ vanishes on Borel sets of Hausdorff-dimension $d-1$. Then there exists a convex function $\varphi$ on $\mathbb{R}^d$ whose gradient $\nabla\varphi$ is a measure preserving map, pushing forward $P_\varepsilon$ to $P_Y$. Although $\varphi$ may not be unique, $\nabla\varphi$ is uniquely determined $P_\varepsilon$-almost everywhere.
\end{theorem}
The assumption that $P_\varepsilon$ vanishes on Borel sets of smaller Hausdorff dimension is actually implied by the assumption that $P_\varepsilon$ be absolutely continuous, so that the theorem in particular holds in the absolutely continuous case. The theorem can also be strengthened in the case where $P_Y$ is also absolutely continuous. Under this additional assumption, there exists $\nabla\varphi^*$ such that for $P_\varepsilon$ almost every $e$ and $P_Y$ almost every $y$,
\begin{equation}\label{inversenabla}
\nabla\varphi^*\circ\nabla\varphi(e)=e,\quad\text{and}\quad\nabla\varphi\circ \nabla \varphi^*(y)=y,
\end{equation} and $\nabla\varphi^*$ is the $P_Y$ almost everywhere unique gradient of a convex function which pushes forward $P_Y$ to $P_\varepsilon$ \citep[p.~67]{villani2003topics}. $\varphi^*$ is the \emph{Legendre-Fenchel transform} of $\varphi$, see \citet[\S12]{rockafellar1997convex}. This can be expressed in a more intuitive way as
\begin{equation}
T^{-1}=(\nabla\varphi)^{-1}=\nabla\varphi^*,
\end{equation}
i.e.~the Brenier map $T$ is invertible almost everywhere on its domain and its inverse is given by $\nabla\varphi^*$.

Gradients of convex functions are so intriguing for Econometricians, because they define a notion of monotonicity which is often argued to be the most natural generalization of a monotone function to higher dimensions. $T=\nabla\varphi$ is monotone in the following sense \citep[p.~53]{villani2003topics}:
\[\langle \nabla\varphi(x)-\nabla\varphi(z),x-z\rangle\geq 0,\] where $\langle\cdot,\cdot\rangle$ denotes the inner product on $\mathbb{R}^d$. Here it is easy to see that if $x>z$ in the partial ordering induced by the positive cone on $\mathbb{R}^d$, then this definition implies that $\nabla\varphi(x)\geq\nabla\varphi(z)$ or that $\nabla\varphi(x)$ and $\nabla\varphi(z)$ are not comparable. In particular, $\nabla\varphi$ \emph{does not involve crossings} in the sense that for $x\in\mathbb{R}^d$
\begin{equation}
(1-t)x+t\nabla\varphi(x)=(1-t)x'+t\nabla\varphi(x')\quad \text{implies that}\quad x=x'\thickspace\text{for}\thickspace t\in[0,1),
\end{equation} 
see the discussion in \cite{mccann1997convexity}; we use this property in the proof of Lemma \ref{secondlemmatheorem2} below. 
Lastly, note that Theorem \ref{mccann} also holds for $P_{Y|X=x}$, by disintegration. For a readable introduction to the theory of the Brenier map consider \citet[Chapter 2]{villani2003topics}. 

\section{Proofs from section \ref{mainidentsection}}

\subsection{Proof of Theorem \ref{maintheorem2}}
Here we prove Theorem \ref{maintheorem2} and its analogue for the case where $Z$ is absolutely continuous.
We need the following important lemmata for the proof.
\begin{lemma}\label{brenierismetricprojection}
Let $F_{X|Z=z}$ and $F_{X|Z=z'}$ be continuous and (multivariate) strictly increasing with the same support $\mathcal{X}_z=\mathcal{X}_{z'}$ and let Assumption \ref{supportass} hold. Let $T\coloneqq\nabla\varphi$ be the Brenier map between $F_{X|Z=z}$ and $F_{X|Z=z'}$. Then, for each $x^*\in\mathcal{X}_z=\mathcal{X}_{z'}$, $T$ is either the metric projection of $x^*$ onto the epigraph of the isoquant $I_{z'}(x^*)$ or its inverse.
\end{lemma}
\begin{proof}
Recall that the graph of $T$, $\Gamma$, is cyclical monotone \citep[Theorem 2.3]{gangbo1996geometry}. That is, for all $m\geq 1$, and for all $(x_1,Tx_1),\ldots,(x_m,Tx_m)\in\Gamma$,
\[\sum_{i=1}^m\|x_i-Tx_i\|_2^2\leq\sum_{i=1}^m\|x_i-Tx_{i-1}\|_2^2,\] with the convention $x_0=x_m$. 

Now by assumption, the epigraphs $\text{epi}(F_{X|Z=z};\alpha)$ and $\text{epi}(F_{X|Z=z'};\alpha)$ are closed convex subsets of $\mathcal{X}_z=\mathcal{X}_{z'}$ for all $\alpha\in[0,1]$. Focusing now on some $\alpha\in[0,1]$ and a corresponding $x^*\in\mathcal{X}_z$ such that the corresponding two isoquants $I_z(x^*)=\{x\in\mathcal{X}_z:F_{X|Z=z}(x)=\alpha\}$ and $I_{z'}(x^*)=\{x\in\mathcal{X}_{z'}:F_{X|Z=z'}(x)=\alpha\}$, we know that if they interesect, they do so either by intersecting transversally---which by the transversality theorem yields a closed manifold \citep[p.~43]{mas1989theory}---or coincide, which trivially results in a closed manifold. Also recall that $T^{-1}\coloneqq\nabla\varphi^*$ is the Brenier map between $F_{X|Z=z'}$ and $F_{X|Z=z}$. With this, our goal is now to show that for each $x\in I_z(x^*)$ for which $F_{X|Z=z}(x)>F_{X|Z=z'}(x)$, the Brenier map is the \emph{metric projection} of $x$ onto the closed convex set $\text{epi}(I_{z'}(x^*))$. 
The metric projection $T$ of $x$ onto the closed convex set $I_{z'}(x^*)$ maps $x$ onto the point $y\in I_{z'}(x^*)$ which is closest to $x$ in the sense that \[y=\argmin_{z\in I_{z'}(x^*)}\|x-z\|_2^2.\] This map exists and is unique since $I_{z'}(x^*)$ is a closed an convex subset of $\mathcal{X}_z=\mathcal{X}_{z'}$ \citep[p.~248]{aliprantis2006infinite}. By our assumption on the intersection of the manifolds, there are two forms this projection can take: in the case of transversal intersection, it is non-trivial in the sense that the distance between $x$ and its metric projection is positive since $I_{z}(x^*)\cap I_{z'}(x^*)$ is a manifold by the transversality theorem, and since $x^*\in I_{z}(x^*)\cap I_{z'}(x^*)$, so that all those $x$ for which $F_{X|Z=z}(x)>F_{X|Z=z'}(x)$ lie outside $\text{epi}(I_{z'}(x^*))$. In the case where the two manifolds coincide the projection is trivial in the sense that the metric projection of $x$ onto $\text{epi}(I_{z'}(x^*))$ coincides with $x$, because $x$ by definition must lie on the boundary of both $I_{z}(x^*)$ and $I_{z'}(x^*)$. Moreover, since $\mathcal{X}_z=\mathcal{X}_{z'}$, those $x$ are in $\mathcal{X}_z$ just as $\text{epi}(I_{z'}(x^*))$ is, so that there does indeed exist a projection of those $x$ onto $\text{epi}(I_{z'}(x^*))$ in both cases.
Now since both $F_{X|Z=z}$ and $F_{X|Z=z'}$ are absolutely continuous, it follows that if a map $T$ is cyclically monotonic, then it must be the Brenier map \citep[p.~80]{villani2003topics}. Therefore, we only have to show that the metric projection satisfies cyclic monotonicity. So for $m\geq 1$ pick $(x_i,Tx_i)$, $i=1,\ldots,m$, on $I_z(x^*)$ for which 
\[F_{X|Z=z}(x_1)>F_{X|Z=z'}(x_1)\quad\text{and}\quad F_{X|Z=z}(x_2)>F_{X|Z=z'}(x_2),\] and where $T$ is the metric projection. Then by definition we have \[\|x_i-Tx_i\|_2^2<\|x_i-y_i\|_2^2\quad\text{for all}\thickspace\medspace y_i\in\text{epi}(I_{z'}(x^*))\thickspace\medspace\text{and all}\thickspace\medspace i=1,\ldots,m,\] but this immediately implies cyclic monotonicity of $T$ for those points. 
Another way to see that the Brenier map between $F_{X|Z=z}$ and $F_{X|Z=z'}$ is the metric projection for those $x\in I_z(x^*)$ which lie outside $\text{epi}(I_{z'}(x^*))$ is by noting that the metric projection onto a closed convex set in a Hilbert space is the gradient of a convex function (\citeauthor{moreau1965proximite} \citeyear{moreau1965proximite}, \citeauthor{holmes1973smoothness} \citeyear{holmes1973smoothness}), and by Theorem \ref{mccann} this gradient of a convex function is the unique measure preserving map between $F_{X|Z=z}$ and $F_{X|Z=z'}$, i.e.~the Brenier map.

Since both $F_{X|Z=z}$ and $F_{X|Z=z'}$ are absolutely continuous, the analogous argument must hold for points $x\in I_{z'}(x^*)$ where $F_{X|Z=z}(x)<F_{X|Z=z'}(x)$, only for $T^{-1}$. Recall that when both distributions are absolutely continuous, the Brenier map is almost everywhere invertible with inverse $T^{-1}=\nabla\varphi^*$, where $\varphi^*$ is the convex conjugate of $\varphi$. Now since $I_z(x^*)$ and $I_{z'}(x^*)$ intersect transversally, just like above it holds that those points $x$ lie outside $\text{epi}(I_z(x^*))$. From the same reasoning it follows that $T^{-1}$ cannot be the Brenier map unless is the metric projection of $x\in I_{z'}(x^*)$ onto the convex set $\text{epi}(I_z(x^*))$. Therefore, the Brenier map $T$ for those $x$ is the inverse of the metric projection onto $\text{epi}(I_z(x^*))$. This shows that $T(I_z(x^*)\cap I_{z'}(x^*))=I_z(x^*)\cap I_{z'}(x^*)$. Having established this, it follows immediately that for each $x\in I_z(x^*)\cap I_{z'}(x^*)$ it actually holds that $Tx=x$ for every point $x\in I_z(x^*)\cap I_{z'}(x^*)$ if $T$ is the metric projection, which must coincide with the Brenier map by uniqueness. 
\end{proof}
\begin{lemma}\label{firstlemmatheorem2}
Let Assumption \ref{supportass} hold and let $T\coloneqq\nabla\varphi$ be the Brenier map between the two distribution functions $F_{X|Z=z}$ and $F_{X|Z=z'}$. If $0<F_{X|Z=z}(x^*)=F_{X|Z=z'}(x^*)<1$ for some $x^*\in\mathcal{X}_z=\mathcal{X}_{z'}$ then $Tx^*=x^*$.
\end{lemma}
\begin{proof}
First note that
\[F_{X|Z=z}(x^*)=F_{X|Z=z'}(x^*)\Leftrightarrow P_{X|Z=z}(X\leq x^*)=P_{X|Z=z'}(X\leq x^*),\] so that $Tx^*=x^*$ would be measure preserving and hence an admissible solution. The Brenier map in this setting is the metric projection onto the epigraph of the respective isoquant by Lemma \ref{brenierismetricprojection}. But for each point which lies on the intersection between two epigraphs the metric projection is the point itself, so that $Tx^*=x^*$.
\end{proof}

\begin{lemma}\label{secondlemmatheorem2}
Let Assumption \ref{supportass} hold and consider the Brenier map $T\coloneqq \nabla\varphi$ between $F_{X|Z=z}$ and $F_{X|Z=z'}$. If at $x_0\in\mathcal{X}_{z}=\mathcal{X}_{z'}$ it holds that $F_{X|Z=z}(x_0)>F_{X|Z=z'}(x_0)>0$, then it must hold that $F_{X|Z=z}(Tx_0)\geq F_{X|Z=z'}(Tx_0)$.
\end{lemma}
\begin{proof}
Pick the corresponding $\alpha\in[0,1]$ for which $F_{X|Z=z}(x_0)=\alpha$ and assume without loss of generality that $x_0$ lies outside the closed convex $\text{epi}(F_{X|Z=z};\alpha)$. Moreover, suppose by contradiction that it holds $F_{X|Z=z}(x_0)>F_{X|Z=z'}(x_0)$ but $F_{X|Z=z}(Tx_0)<F_{X|Z=z'}(Tx_0)$. Since the isoquants are all convex and all distribution functions are strictly increasing, it follows that $Tx_0$ must lie inside $\text{epi}(F_{X|Z=z};\alpha)$. But then $T$ cannot be the metric projection of $x_0$ onto this epigraph as in Lemma \ref{brenierismetricprojection}, a contradiction.
\end{proof}

The following lemma proves the intuitive explanation from page \pageref{equalitiesimportant} formally. 
\begin{lemma}\label{stringofequalitieslemma}
The first equality in \eqref{equalitiesimportant} holds.
\end{lemma}
\begin{proof}
The simplest proof is to consider $P_{\varepsilon|X=x,Z=z}$ and $P_{\varepsilon|V=h^{-1}(x,z),Z=z}$ to be disintegrations \citep[Chapter 10]{bogachev2007measure2}. Note that these disintegrations exist and coincide with the abstract conditional expectations, because we work in Euclidean space and with absolutely continuous distribution functions \citep[Theorem 1]{chang1997conditioning}. The map $\phi: (\varepsilon,X,Z)\mapsto(\varepsilon,h^{-1}(X,Z),Z)$ is a measure preserving isomorphism since $h^{-1}$ is a measure preserving isomorphism for all $z\in\mathcal{Z}$, just as the identity maps $\varepsilon\mapsto\varepsilon$ and $Z\mapsto Z$. But then by Corollary 5.24 in \citet{einsiedler2013ergodic} the disintegrated measure $P_{\varepsilon|V=h^{-1}(x,z),Z=z}$ is the pushforward of the disintegrated measure $P_{\varepsilon|X=x,Z=z}$, that is
\[P_{\varepsilon|V=h^{-1}(x,z),Z=z}(E) = P_{\varepsilon|\phi(x,z)}(E) = P_{\varepsilon|X=x,Z=z}(\phi^{-1}(E)) = P_{\varepsilon|X=x,Z=z}(E),\] for all $E\in\mathscr{B}_{\mathbb{R}^d}$. Here, the first equality follows by the definition of $\phi$, the second equality follows by the definition of a pushforward measure as in Corollary 5.24 of \citet{einsiedler2013ergodic}, and the third equality follows from the fact that $\phi$ maps $\varepsilon$ to $\varepsilon$.
\end{proof}
We can now prove the main theorem.

\paragraph{Proof of Theorem \ref{maintheorem2}}
\begin{proof}
Recall that we want to show identification, so that the assume there is $m:\mathcal{E}\to\mathcal{Y}_x$ as well as $m^*:\mathcal{E}\to\mathcal{Y}_x$ with corresponding unobservable distributions $F_\varepsilon$ and $F_{\varepsilon^*}$ so that $(m,\varepsilon)$ and $(m^*,\varepsilon^*)$ are equivalent in the sense that they generate the same distribution $F_{Y|X=x}$ for endogenous $X$.
We work on a suitably equipped complete probability space on which $\varepsilon$, $\varepsilon^*$, $X$, $Z$, and $U$ are defined. Let us first prove that the identified set is 
\[\mathcal{I}\coloneqq \{m\in\mathcal{H}(Y_x,\varepsilon):(m^{-1}(X,Y),U)\independent Z\}.\] The proof of this is almost exactly like the original proof in \cite{torgovitsky2015identification}. If $m$ is in the identified set, then $Y=m(X,\varepsilon)$ for some $\varepsilon$ with $(\varepsilon,U)\independent Z$. Now since $m$ is a measure preserving isomorphism we have $\varepsilon = m^{-1}(X,Y)$. It therefore holds that $(m^{-1}(X,Y),U)\independent Z$, so that $m\in\mathcal{I}$. On the other hand, if $m\in\mathcal{I}$, then $(m^{-1}(X,Y),U)\independent Z$ and since $\varepsilon=m^{-1}(X,Y)$ $m$ is in the identified set. Let us now prove the main part of the theorem.

Since we want to show identification, the goal is to show that $m^*=m$ almost everywhere as in Figure \ref{isostructure}, so that $\mathbb{I}$ contains a unique $m$.
We can assume that $Z$ is binary with points $z,z'\in\mathcal{Z}$ and that $F_{X|Z=z}$ and $F_{X|Z=z'}$ only intersect in one manifold $\mathcal{M}(z,z')$. In fact, if there are more manifolds, it does not matter which one we use; moreover, since the manifolds are of measure zero, their union is still of measure zero, so that we still can identify $q$ modulus the measure zero set of the union of the manifolds.

Now since both $m$ and $m^*$ are isomorphisms and continuous in $x$, $q(x,z,\cdot)$ must be an isomorphism and continuous in $x$ for all $(x,z)\in\mathcal{X}\times\mathcal{Z}$. Recall that by Assumption \ref{instrumentass} $q$ does not depend on $z$, so that we can write $q(x,z,\cdot)=q(x,\cdot)$. Furthermore, note that 
\[q(x,\varepsilon) = m^{-1}(x,m^*(x,\varepsilon))\quad\text{for all}\thickspace\medspace (x,z,\varepsilon)\in\mathcal{X}\times\mathcal{Z}\times\mathcal{E}.\] Therefore, $q(x,\varepsilon)=\varepsilon$ if and only if $m^*=m$. 
Now by the assumed normalization, as soon as $q(x,\varepsilon)$ is a function which does not depend on $x$, we know that it must hold that $q(x,\varepsilon)=\varepsilon$, based on the reasoning on page \pageref{explanation}; as mentioned, this is analogous to the reasoning in \citet{torgovitsky2015identification}, only for measure preserving isomorphisms instead of quantile functions. Therefore, we only need to prove that $q(x,\cdot)$ is independent of $x$.

Now by Assumption \ref{normalizationass} $h$ is the identity for $z$ and the Brenier map for $z'$, so that $F_U=F_{X|Z=z}$. Moreover, the fact that $h(\cdot,z)$ and $h(\cdot,z')$ are measure preserving isomorphisms implies that conditioning on the event $[X=x,Z=z]$ is the same as conditioning on the event $[V=g^{-1}(x,z),Z=z]$ by Lemma \ref{stringofequalitieslemma}.

Now in order to prove that $q$ is constant for almost all $x$, we use the general sequencing argument mentioned in the main text, similar in spirit to \cite{d2015identification} and \cite{torgovitsky2015identification}, only for several dimensions. To make the notation simpler in the following, we denote the Brenier map $h(z',u)$ mapping $F_{X|Z=z'}$ onto $F_{X|Z=z}=F_U$ by $T$, and its inverse $h^{-1}(x,z')$ by $T^{-1}$. We thus want to show that for each $x\in\mathcal{X}_z=\mathcal{X}_{z'}$ either $\lim_{n\to\infty} T^nx=x_m$ or $\lim_{n\to\infty}(T^{-1})^nx=x_m$ for some $x_m\in\mathcal{M}(z,z')$, so that $q(x,\cdot)$ is independent of $x$ up to the measure zero set $\mathcal{M}(z,z')$.\footnote{We denote by $T^n$ the $n$-fold repeated application of $T$, i.e.~$T^n(x) \equiv T(T(\ldots T(x))\ldots)$.} In words, the idea is, for every $x\in\mathcal{X}_z=\mathcal{X}_{z'}$, to obtain the sequence which does not change $F_\varepsilon$, but changes $F_{X|Z=z}$, so that we can obtain the exogenous effect of $X$ on the second stage. As stated in the main text, we need to show that this sequence converges, and we will now show that it converges to some element $x_m\in\mathcal{M}(z,z')$. This will be the difficult step in what follows.

So pick some $x\in\mathcal{X}_z=\mathcal{X}_{z'}$. If $F_{X|Z=z}(x) = F_{X|Z=z'}(x)$ there is nothing to prove as by Lemma \ref{firstlemmatheorem2} it holds that $T^nx=x$ for $n\in\mathbb{N}$, so that this $x$ already has converged and must lie in $\mathcal{M}(z,z')$ since we assume there is only one manifold where the distribution functions intersect. Now depending on the interplay between $\mathcal{M}(z,z')$ and $x$ as well as $F_{X|Z=z}(x)$ and $F_{X|Z=z'}(x)$, we have to decide if we iterate $T$ or $T^{-1}$ to find convergence. The underlying idea is that we always iterate such that the sequence converges towards $\mathcal{M}$. Now since by part 4 (ii) of Assumption \ref{supportass} all elements $x$ for which $F_{X|Z=z}(x)=0=F_{X|Z=z'}(x)$ lie on one side of the manifold, which means that for $x_0$ with $1>F_{X|Z=z}(x_0),F_{X|Z=z'}(x_0)>\varepsilon$ for some $\varepsilon>0$ it must either hold that $F_{X|Z=z}(x_0)> F_{X|Z=z'}(x_0)$ or the reverse strict inequality; note that we can rule out the case where equality holds since then this point lies on some manifold and has already converged as just mentioned. So for those points with $F_{X|Z=z'}(x_0)>F_{X|Z=z}(x_0)$, we iterate $T$, i.e.~from $F_{X|Z=z'}$ to $F_{X|Z=z}$. The reason is the same as in the one-dimensional case as depicted in the following figure.
\begin{figure}[h!t]
\centering
\begin{tikzpicture}
\draw[->, thick] (0,0) to (0,4);
\draw[->,thick] (0,0) to (4,0);
\draw[-,thick] (0,0) to [out=15,in=180] (3.7,3.7);
\draw[thick] (0,0) to [out=90,in=220] (3.7,3.7);
\draw[-,thick] (0,3.7) to (-0.1,3.7);
\draw node[left] at (0,3.7) {$1$};
\draw node[right] at (4,0) {$x$};
\draw[-,thick] (0.5,0) to (0.5,-0.1);
\draw node[below] at (0.5,-0.1) {$x_0$};
\draw[-,dashed] (0.5,0.2) to (0.5,1.47);
\draw[-,thick] (0,0.2) to (-0.1,0.2);
\draw node[left] at (-0.1,0.2) {$F_{X|Z=z}(x_0)$};
\draw[-,dashed] (0.52,1.5) to (1.5,1.5);
\draw[-,thick] (1.5,0) to (1.5,-0.1);
\draw node[left] at (-0.1,1.5) {$F_{X|Z=z}(Tx_0)=F_{X|Z=z'}(x_0)$};
\draw node[below] at (1.5,-0.1) {$Tx_0$};
\draw[-,thick] (0,1.5) to (-0.1,1.5);
\draw[-,dashed] (1.5,1.5) to (1.5,2.3);
\draw[-,dashed] (1.5,2.3) to (1.9,2.3);
\draw[-,dashed] (1.9,2.3) to (1.9,2.55);
\draw[-,dashed] (1.9,2.55) to (2.05,2.55);
\draw[-,thick] (0,2.3) to (-0.1,2.3);
\draw node[left] at (-0.1,2.3) {$F_{X|Z=z}(T^2x_0)=F_{X|Z=z'}(Tx_0)$};
\draw[->,thick] (3,1.9) to (2.2,2.6);
\draw node[right] at (3,1.8) {$\mathcal{M}(z,z')$};
\draw node[right] at (2.5,0.5) {$F_{X|Z=z}$};
\draw[->,thick] (2.5,0.6) to (1.43,1.17);
\draw node[right] at (4.2,2.6) {$F_{X|Z=z'}$};
\draw[->,thick] (4.2,2.7) to (3.3,3.3);
\end{tikzpicture}

\caption{Sequence for a point $x_0$ in the one-dimensional case}
\label{sequenceexample}
\end{figure}
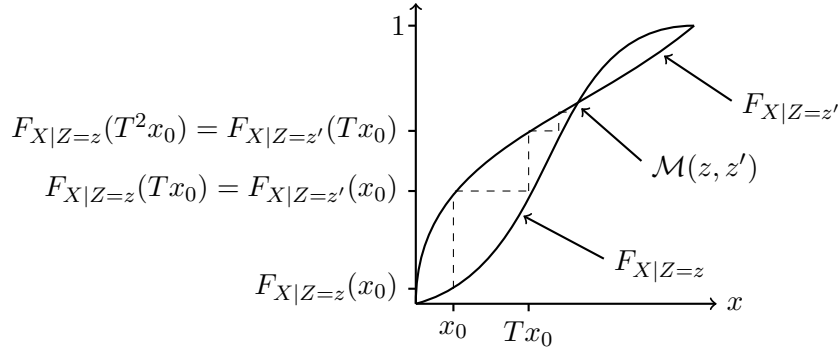

For the reverse inequality we need to iterate $T^{-1}$ from $F_{X|Z=z}$ to $F_{X|Z=z'}$.
Therefore, without loss of generality assume $F_{X|Z=z'}(x)>F_{X|Z=z}(x)$. Then iterating $T$ converges to some $x_m\in\mathcal{M}(z,z')$, because we have
\[F_{X|Z=z'}(T^{n+1}x)\geq F_{X|Z=z}(T^{n}x)\] by Lemma \ref{secondlemmatheorem2}. Hence, taking the limit $n\to\infty$ by absolute continuity of the two distribution functions yields
\[F_{X|Z=z'}(\lim_{n\to\infty}T^{n}x)=\lim_{n\to\infty} F_{X|Z=z}(T^{n}x),\] which is either a monotonically increasing sequence bounded above by some $x_m\in\mathcal{M}(z,z')$ with $0<F_{X|Z=z}(x_m)=F_{X|Z=z'}(x_m)<1$ or a monotonic decreasing sequence bounded below by some $x_m\in\mathcal{M}(z,z')$, depending on which side of $\mathcal{M}(z,z')$ $x$ lies; since those sequences are monotone and bounded in both cases they must converge to $x_m\in\mathcal{M}(z,z')$. This reasoning works in particular, because the range of Brenier maps $T$ and $T^{-1}$ are the respective supports of $F_{X|Z=z}$ and $F_{X|Z=z'}$, so that $T$ never maps outside $\mathcal{X}_z=\mathcal{X}_{z'}$.

Finally, if there are more manifolds than just $\mathcal{M}(z,z')$, then for some $x$ the respective iterations might converge to different manifolds, but still all elements will converge to \emph{some} manifold, because all $x$ for which the iterations do not converge to some other manifold, they must converge to $\mathcal{M}(z,z')$ by the above argument. Since all manifolds are of measure zero by assumption, their (finite) union must also be of measure zero.

We are now in the position to conclude the proof by applying the sequencing argument to show that $q(x,e)$ does not depend on almost every $x$. So pick some $x_0\in\mathcal{X}_z=\mathcal{X}_{z'}$; we assume that $0<F_{X|Z=z}(x_0)<F_{X|Z=z'}(x_0)$ for $x_0$ lying on the same side of the manifold as the elements which have zero probability, but the reasoning in the other cases is analogous. By the independence $(\varepsilon,U)\independent Z$ as well as Lemma \ref{stringofequalitieslemma} and the fact that $h$ is a measure preserving isomorphism, we have for every Borel set in $\mathscr{B}_{\mathcal{E}_x}$
\begin{align*}
P_{\varepsilon^m|X=x_0,Z=z}(E) &=P_{\varepsilon^m|U=h^{-1}(x_0,z),Z=z}(E) \\
&= P_{\varepsilon^m|U=h^{-1}(Tx_0,z'),Z=z}(E)  \\
&= P_{\varepsilon^m|X=Tx_0,Z=z'}(E)=P_{\varepsilon^m|X=Tx_0,Z=z}(E),
\end{align*}
for every Borel set $E\in\mathscr{B}_{\varepsilon^m|X}$; analogously, 
\begin{align*}
P_{\varepsilon|X=x_0,Z=z}(E') &=P_{\varepsilon|U=h^{-1}(x_0,z),Z=z}(E') \\
&= P_{\varepsilon|U=h^{-1}(Tx_0,z'),Z=z}(E')  \\
&= P_{\varepsilon|X=Tx_0,Z=z'}(E')=P_{\varepsilon|X=Tx_0,Z=z}(E').
\end{align*}
for every Borel set $E'\in\mathscr{B}_{\mathcal{E}_{x_0}}$. 

Now to conclude the proof of Theorem \ref{maintheorem2} recall that $q$ is a measure preserving isomorphism mapping each $E'\in\mathscr{B}_{\mathcal{E}_{x_0}}$ onto a unique $q(x_0,E')$. Thus as 
$P_{\varepsilon|X=x_0,Z=z}(E') = P_{\varepsilon|X=Tx_0,Z=z}(E'),$ it must hold that $q(x_0,E')  =  q(Tx_0,E')$. This holds if we iterate $T$, so that
\[q(T^nx_0,E') = \cdots = q(x_0,E')\quad\text{for all}\thickspace\medspace E'\in\mathscr{B}_{\varepsilon|X}\thickspace\thickspace\text{and}\thickspace n\in\mathbb{N},\] since $\mathcal{E}_{x,z}=\mathcal{E}$ by Assumption \ref{supportepsilon}. Now from above we know that the iteration $T^nx_0$ converges to some element $x_m\in\mathcal{M}(z,z')$ and since $q(\cdot,E')$ is continuous in $x$, we have
\begin{equation}\label{continuityequation}
q(x_m,E') = q(\lim_{n\to\infty}T^nx_0,E') = \lim_{n\to\infty}q(T^nx_0,E')=\lim_{n\to\infty}q(x_0,E')=q(x_0,E').
\end{equation} 
Therefore $q(\cdot,E')$ is constant on $\mathcal{X}_z=\mathcal{X}_{z'}$ modulus $\mathcal{M}(z,z')$, which by assumption is of measure zero, which concludes the proof.

Lastly, if Assumption \ref{mpiso2prime} holds in place of Assumption \ref{mpiso2} then \eqref{continuityequation} does not hold. However, by definition it holds for all $\delta>0$ that 
\[\lim_{n\to\infty} P_{\varepsilon|X=x_n,Z=z}(e: |q(T^nx_0,e)-q(x_0,e)|>\delta) = 0\quad\forall e\in\mathcal{E}_{xz}.\]
This implies that $q(x_m,\cdot)=q(x_0,\cdot)$ up to a set of measure zero in $\mathcal{E}_{xz}$, because it is a well-known result that if a sequence $\{f_n\}_{n\in\mathbb{N}}$ converges to $f$ and to $g$ in measure, then $f$ coincides with $g$ almost everywhere. This means that $q$ is constant on $\mathcal{X}_z=\mathcal{X}_{z'}$ modulus $\mathcal{M}(z,z')$ for almost every $\varepsilon$, so that $m$ can be identified for almost every $x$ and almost every $\varepsilon$.
\end{proof}

\subsection{Identification result for absolutely continuous $Z$}
In the not-so-important case where $Z$ is absolutely continuous we can eliminate Assumption \ref{supportass}, but have to assume that $m$ is differentiable in $x$. The statement and proof of this theorem are analogous to the result in \citet{torgovitsky2015supplement}, the only difference being that we allow for $m$ to be a (possibly multivariate) measure preserving $C^{1}$-diffeomorphism. We still state the whole proof for the convenience of the reader. We denote by $d_x$ the dimension of $X$ and by $d_z$ the dimension of $Z$.
\begin{theorem}[Identification of nonseparable triangular models with absolutely continuous $Z$]\label{maintheorem2abscont}
Suppose that $\mathcal{X}$ is convex, $Z$ is a continuously distributed vector-valued random variable satisfying Assumption \ref{instrumentass}, and $m(x,z,\varepsilon):\mathcal{E}_z\to\mathcal{Y}_{xz}$ is a determinable measure preserving isomorphism differentiable in all $x$ and $z$ satisfying Assumption \ref{mpiso2}. Moreover, let Assumptions \ref{mpiso}, \ref{abscont}, and \ref{supportepsilon} hold. Denote by $G(x,z)$ the $d_z\times d_x$ matrix with $(j,k)$ element $\nabla_{z_j}g^{-1}(x_k,z)$. Then $m$ is point-identified on $\mathcal{X}\times\mathcal{E}$ if for every $x$ in a dense subset $\mathcal{X}_d$ of $\mathcal{X}$ and every $y\in\mathcal{Y}^\circ_x$, there exists a $\bar{z}$ with $x\in\mathcal{X}^\circ_{\bar{z}}$ and $y\in\mathcal{Y}^\circ_{x,\bar{z}}$ for which $\nabla_xg^{-1}(x,\bar{z})$ exists and $G(x,\bar{z})$ exists and has rank $d_x$.
\end{theorem}
\begin{proof}
The proof is analogous to the proof in \citet{torgovitsky2015supplement}. The ultimate goal is to show that $q(x,z,e)$ from Figure \ref{isostructure} does not depend on $x$ by differentiation. First, $q(x,z,e)$ is $C^{1}$ since $m^*$ and $m$ by assumption are $C^{1}$-diffeomorphisms. Then analogous to the proof of Theorem \ref{maintheorem2}, we have that $q(x,z,e)=q(x,e)$ since $Z$ is a valid instrument. Now conditioning on the event $[U=h^{-1}(x,z),Z=z]$ is the same as conditioning on the event $[X=x,Z=z]$ by Lemma \ref{stringofequalitieslemma}. Therefore, by construction we have 
\[q(x,e)= q(x,z,e) = q(u,e) = q(h^{-1}(x,z),e)\] as depicted in Figure \ref{isostructure}. Now fix some $x\in\mathcal{X}$ and $e\in\mathcal{E}_{xz}$ so that $y=m(x,e)\in\mathcal{Y}^\circ_{xz}$ and take $\bar{z}$ as in the statement of the theorem. Then $e\in\mathcal{E}^\circ_{x,\bar{z}}$ as $y\in\mathcal{Y}^\circ_{x,\bar{z}}$ so that we can differentiate $q(x,z,e)$ in a neighborhood of $(x,\bar{z},e)$. Now let us differentiate $q$ first with respect to $z_j$. This gives by the chain rule
\[0 = \nabla_{z_j} q(x,e) = G_j(x,\bar{z})\nabla_1 q(g^{-1}(x,\bar{z}),e)' \]
where $G_j(x,\bar{z})$ is the $j$-th row of $G(x,\bar{z})$ and $\nabla_1$ denotes differentiation with respect to the first argument of $q(g(v,\bar{z}),e)'$, so that $\nabla_1 q(g(v,\bar{z}),e)'$ is a $d_x$ column vector. Stacking the vectors for each $j$ gives
\[G(x,\bar{z})\nabla_1q(g^{-1}(x,\bar{z}),e)' = 0_{d_z}.\]
By the assumption that $G(x,\bar{z})$ has full rank, this implies $\nabla_1q(g^{-1}(x,\bar{z}),e)'=0_{d_x}$. Using this information, we can now finalize the proof by differentiating $q(x,e)$ with respect to $x_k$ at $(x,\bar{z},e)$. In fact, since $\nabla_1q(g^{-1}(x,\bar{z}),e)'=0_{d_x}$ this implies by the chain rule that
\[\nabla_{x_k} q(x,e) = \nabla_{x_k} g^{-1}(x_k,\bar{z})\nabla_{1,k}q(g^{-1}(x,\bar{z}),e) = 0\] for every element $k$. Here, $\nabla_{1,k}$ denotes differentiation of the first argument with respect to the $k$-th element of $x$. But this means that $\nabla_{x} q(x,e)=0_{d_x}$ for all $x\in\mathcal{X}$ as $\nabla_x q(x,e)$ is continuous and $\mathcal{X}_d$ is dense in $\mathcal{X}$ as argued in \citet{torgovitsky2015supplement}. By convexity of $\mathcal{X}$ this implies that $q(x,e)=r(e)$ for some function $r$ for every $x\in\mathcal{X}$ and all $e\in\mathcal{E}_x$ by continuity. This shows that $q(x,e)$ does not depend on $x$ and by the reasoning from the main text this means that $m$ is identified.
\end{proof}
Note that identification here is for all elements in the support, not merely almost all elements. 

\section{Proof of Theorem \ref{utilityidentprop}}
The following definition of weak convergence of probability measures is needed, which we have taken from Definition 8.1.1 in \citet{bogachev2007measure2}.
\begin{definition}
A sequence of probability measures $\{P_n\}_{n\in\mathbb{N}}$ on a measurable space $(\mathcal{X},\mathscr{A}_X)$ is weakly convergent to a probability measure $P$ if for every bounded and continuous real function $f$ on $\mathcal{X}$ one has
\[\lim_{n\to\infty}\int_{\mathcal{X}}f(x)P_n(dx) = \int_{\mathcal{X}}f(x)P(dx).\]
We denote this convergence by $P_n \Rightarrow P$.
\end{definition}

Moreover, we need the following strengthening of Corollary 5.23 in \citet{villani2008optimal}.
\begin{lemma}\label{villanistrengthening}
Let $\mathcal{X}$ and $\mathcal{Y}$ be open subsets of $\mathbb{R}^n$ and let $c:\mathcal{X}\times \mathcal{Y} \to\mathbb{R}$ be a continuous cost function with $\inf c>-\infty$. Let $\{P_n^X\}$ and $\{P_n^Y\}$ be sequences of probability measures on $\mathcal{X}$ and $\mathcal{Y}$, respectively, such that $P_n^X\leq C P_X$ for all $n\in\mathbb{N}$ and some $0<C<\infty$. Furthermore, $P_n^X$ converges weakly to $P_X$ and $P_n^Y$ converges weakly to $P_n^Y$. For each $n$ let $\pi_n$ be the optimal transference plan (i.e.~the respective solution of the Kantorovich problem) between $P_n^X$ and $P_n^Y$. Furthermore, assume that 
\[\int c \pi_n<\infty\thickspace\forall n\quad\text{and}\quad\lim_{n\to\infty}\int c\pi_n<\infty.\]
Suppose also that the corresponding Monge problems have a unique solution and that the optimal transport plans $\pi_n$ and $\pi$ are concentrated on the graph of the optimal transport maps $T_n$ and $T$ solving the respective Monge problems. Then 
\[\lim_{n\to\infty} P_n^X(x: |T_n(x)-T(x)|\geq \varepsilon)=0 \quad\text{for all}\thickspace\varepsilon>0.\]
\end{lemma}
The proof is verbatim the proof of Proposition 50 in \citet{lindsey2016optimal}. The idea is to require stronger assumptions on the cost function $c$ to generalize Corollary 5.23 in \citet{villani2008optimal}, which only requires $c$ to be lower semi-continuous, but can only guarantee the result for a fixed $P_X$ and a convergent sequence $P_n^Y$. Since in our case both measures $P_{Y|X=x,Z=z_i}$ and $P_{\varepsilon|X=x,Z=z_i}$ change with $x$, we need this relaxation which allows both sets of measures to drift. 

We are now in the position to prove the important lemma which shows that $m^{-1}(x,y)$ is continuous in measure which we need for our application of Theorem \ref{maintheorem2}.

\begin{lemma}\label{continuitylemma}
Under Assumption \ref{regularityass}, the determinable measure preserving map transporting $P_{Y|X=x,Z=z_i}$ onto $P_{\varepsilon|X=x,Z=z_i}$ is continuous in measure.
\end{lemma}
\begin{proof}
Recall that $F_{X|Z=z}$ and $F_{X|Z=z'}$ are absolutely continuous. Absolutely continuous distributions satisfy Tjur's property at every point of continuity, which is equivalent to the existence of a unique continuous disintegration $P_{Y,X|Z=z}\mapsto P_{Y|X=x,Z=z}$ and analogous for $z'$ in the sense that for each sequence $x_n\to x$, the corresponding probability measures satisfy $P_{Y|X=x_n,Z=z}\Rightarrow P_{Y|X=x,Z=z}$, see \citet[Lemma 2.7]{ackerman2016computability}. Moreover, Assumption \ref{regularityass} guarantees that $m^{-1}(x,y)$ is a unique and bounded measure preserving map for the Monge-Kantorovich problem under the general cost function $\xi(x,\varepsilon,y)$, which is by definition differentiable and hence continuous. We can hence apply Lemma \ref{villanistrengthening}, which guarantees that for every sequence $\{x_n\}_{n\in\mathbb{N}}\in\mathcal{X}_{z}$ converging to some $x\in\mathcal{X}_{z}$, the corresponding optimal transport map $m^{-1}(x_n,y)$ converges to $m^{-1}(x,z)$ in probability. 
\end{proof}

We can now prove the theorem.

\begin{proof}[Proof of Theorem \ref{utilityidentprop}]
Assumption \ref{regularityass} in conjunction with Lemma \ref{continuitylemma} guarantees that $m^{-1}(x,y)$ is continuous in measure and that all assumptions are satisfied in order to apply Theorem \ref{maintheorem2}. Therefore, $m^{-1}(x,y)$ is identified for almost every $x$ and $y$. To show identification of $\bar{U}(x,y)$ we follow the reasoning in \citet{chernozhukov2014single}. Notice that Assumptions \ref{gammaconvex} and \ref{regularityass} guarantee that $V(x,y)$ is differentiable. The proof for this is verbatim Step 2 in the proof of Theorem 3(1) in \citet{chernozhukov2014single}. Since $V(x,y)$ is differentiable and since the inverse demand function $m^{-1}(x,y)$ is uniquely determined for almost every $x$ and $y$, the first order condition 
\[\nabla_y\xi(x,m^{-1}(x,y),y) = \nabla p(y) - \nabla_y\bar{U}(x,y)\] identifies $\nabla_y\bar{U}(x,y)$ for almost every $x$ and almost every $y$ so that $\bar{U}(x,y)$ is almost everywhere identified up to an additive constant, analogous to the reasoning in \citet{chernozhukov2014single}.
\end{proof}

\end{document}